\documentclass[twocolumn,10pt,letterpaper,superscriptaddress,amsmath,amssymb,prb,
pop,
floatfix,
showkeys]{revtex4-2}

\usepackage{amsthm}
\usepackage{graphicx}       
\usepackage{booktabs}       
\usepackage{dcolumn}        
\usepackage{bm}             
\usepackage{float}          
\usepackage[T1]{fontenc}    
\usepackage{xcolor}         
\usepackage{textcomp,gensymb} 
\usepackage{lineno}         
\usepackage{soul}           
\usepackage{ragged2e}       
\usepackage{esint}
\usepackage{cuted} 
\usepackage{tikz}
\usepackage{tikz-3dplot}
\usepackage{braket}

\usepackage{hyperref}       

\hypersetup{
    colorlinks=true,
    linkcolor=blue,
    filecolor=magenta,
    urlcolor=cyan
}
\newcounter{para}

\newtheorem{theorem}{Theorem}[section]  

\theoremstyle{definition}

\theoremstyle{remark}

\begin{document}

\title{Equivalence class of Emergent Single Weyl Fermion in 3d Topological States: 
\\ gapless superconductors and superfluids Vs chiral fermions}

\author{Gabriel Meyniel} 
\email{gabriel.meyniel@ens.psl.eu}
\affiliation{École Normale Supérieure - PSL, physics department, Paris, France}
\affiliation{Department of Physics and Astronomy, University of British Columbia,
6224 Agricultural Road, Vancouver, BC, V6T 1Z1, Canada}

\author{Fei Zhou}
\email{feizhou@phas.ubc.ca}
\affiliation{Department of Physics and Astronomy, University of British Columbia,
6224 Agricultural Road, Vancouver, BC, V6T 1Z1, Canada}

\date{\today}

\begin{abstract}
In this article, we put forward a practical but generic approach towards constructing a large family of $(3+1)$ dimension lattice models which can naturally lead to a single Weyl cone in the infrared (IR) limit. Our proposal relies on spontaneous charge $U(1)$ symmetry breaking via a pair condensate to evade the usual no-go theorem of a single Weyl cone in a 3d lattice. 
We have explored three concrete paths in this approach, all involving fermionic topological symmetry protected states (SPTs).
Path \textbf{a}) is to push a gapped SPT in a 3d lattice with time-reversal symmetry (or $T$-symmetry) to a gapless topological quantum critical point (tQCP) which involves a minimum change of topologies, i.e. $\delta N_w=2$
where $\delta N_w$ is the change of winding numbers across the tQCP. Path \textbf{b}) is
to peel off excessive degrees of freedom in the gapped SPT via applying $T$-symmetry breaking fields which naturally result in a pair of gapless nodal points of real fermions. Path \textbf{c}) is a hybrid of \textbf{a}) and \textbf{b}) where tQCPs, with $\delta N_w \geq 2$, are further subject to time-reversal-symmetry breaking actions.
We identify detailed phase boundaries of emergent single Weyl fermions in these lattice models. 
The whole family of lattice models with single Weyl fermions here can be effectively encoded in two dual copies of (1, 1) representations of a Spin(4) group.
In the infrared limit, all the lattice models with single Weyl fermions studied here are isomorphic to either a tQCP in a DIII class topological superconductor with a protecting $T$-symmetry, or its dual, a $T$-symmetry breaking superconducting nodal point phase, and therefore form an equivalence class.
For a generic $T$-symmetric tQCP along Path \textbf{a}), the conserved-charge operators span a six-dimensional linear space while for a $T$-symmetry breaking gapless state along Path \textbf{b}), \textbf{c}), charge operators typically span a two-dimensional linear space instead.
Finally, we pinpoint connections between three spatial dimensional lattice chiral fermion models and gapless real fermions that can naturally appear in superfluids or superconductors studied previously. 

\end{abstract}

\maketitle

\bigskip

\section{Introduction}

It has been well known that protecting symmetries $G_p$ and resultant Symmetry Protected States (SPT) greatly enlarge the family of topological phases of matter\cite{Chen2012,Schnyder2008,Schnyder2011,Kitaev2009,Fu2007,Moore2007,Qi2008,Qi2009,Hasan2010,Fidkowski2010,Fidkowski2011,Pollman2012,Wang2014,Metlitski2015,Song2017,Fidkowski2013,Wen2013,Bernevig2013,Wen2017}. One of the most obvious consequences of SPTs is the emergence of unique classes of topological quantum phase critical points (tQCPs) in both low and high dimensions \cite{Jiang2018,Verresen2018,Tsui19, Bi2019,Thorngren2021,Xu13}.
These tQCPs signify a change of global topologies, rather than a change of conventional ordering as in a more standard Landau paradigm of order-disorder transitions.
Other striking examples of tQCPs include the ones in two- and three-dimensional topological superconductors where the off-diagonal long range order remains the same while the global topology undergoes an abrupt change due to collapse of standard fully gapped BdG quasi-particles into gapless {\em real} fermions\cite{Yang_Jiang_Zhou,Yang_Zhou,Zhou1,Yang2025}.

A very puzzling and distinct feature of tQCPs in SPTs is perhaps the appearance of emergent symmetries, anomalous symmetries that are uniquely related to gapless tQCPs. These anomalous symmetries can not appear in a gapped SPT with the same protecting symmetry and are non-on site ones. It is highly unique, especially in high-dimensions, e.g. in 3d lattices.

During the last few years, emergent symmetries have also been introduced as enriching new elements in the Landau paradigm of order-disorder transitions\cite{Ji2020,Chatterjee2021}. However, in high-dimensions say 3d lattices, those emergent symmetries typically have higher forms, higher than the standard 0-form and have been noted as generalized symmetries. The emergent symmetries at tQCPs we will utilize below for the purpose of 3d lattice chiral fermions are of the more standard 0-form albeit their UV completion has a few of highly surprising features (See below).

Furthermore, such emergent anomalous symmetries can also appear in a stable high dimensional gapless phase discussed in Ref.\cite{Yang_Zhou,Kapoor,Zhou2}.
Topological stability of those gapless states are much related to nodal 
point phases previously proposed in the contexts of HTc\cite{Balents1998}, spin liquids\cite{Kitaev06,Wen2002}, symmetry protected gapless states
\cite{Beri2010,Kobayashi2014,Schnyder2011,Matsuura2013,Armitage2018}.

It is therefore quite natural to discuss potential applications of such unique emergent symmetries in gapless states, either as tQCPs or as stable gapless phases. Our main objective here is to explore in details one such application, towards an emergent single Weyl fermion in 3d lattices (either left-handed or right-handed) or more precisely lattice chiral fermions via a family of topological gapless states which can be associated to a class of gapped fermonic topological states.

Single Weyl cone dynamics have naturally appeared in quite a few studies of $(3+1)$-dimensional $Z^T_2$-symmetry protected topological states, either as a representation of infrared dynamics at topological quantum critical points (tQCPs)\cite{Zhou2,Zhou3} or as an effective field theory (EFT) of stable nodal point states as a well-known example of gapless topological phases \cite{Yang_Zhou,Kapoor,Zhou1}. In Ref.\cite{Kapoor}, a lattice model was further constructed to verify 
the paradigm of emergent single Weyl fermion by scanning the whole 3-torus of the momentum space.
As a pleasant surprise, the appearance of such dynamics does not contradict the well-known no-go theorem of single Weyl cones in a three-dimensional (spatial) lattice \cite{NIELSEN198120,NIELSEN1981173,Friedan:1982nk,Kiritsis:1986re}, making it possible to develop a lattice model where the ultraviolet (UV) completion of such an IR Weyl cone can be further explored and understood with fine resolutions.

One of the interesting applications of this observation can be the potential emergence of a single copy of supersymmetric (SUSY) conformal field theory (CFT) in a three-dimensional (or two dimensional) bulk with $N_f=\frac{1}{2}$ fermions or one-half of a three dimensional Dirac fermions\cite{Zhou2}. The supersymmetri counterpart of two majorana fermion cones here can form a complex fermion under certain conditions, resulting in an emergent $U(1)$ symmetry.

Two copies of decoupled SUSY CFTs were originally suggested in a two-dimensional lattice\cite{Lee2007}. They were also later suggested to appear in two-dimensional surfaces of three dimensional bulks, where each surface supports one single copy of SUSY CFT\cite{Grover2014,Jian2017,Zerf2016,Fei2016}.

As implied above, in this article we will focus on the UV completion of 
emergent single Weyl fermions, especially the UV completion of emergent symmetries appearing in various gapless topological superconductors and superfluids known to us before\cite{Yang_Zhou,Zhou1,Zhou2,Zhou3,Kapoor}.
The no-go theorem implies that the non-compactness of symmetries can play a very critical role in chiral fermion lattice models\cite{Friedan:1982nk}. The incompatibility between the locality and compactness of symmetry charges has been emphasized as well\cite{Fidkowski23}. 
Very recently this has been further explicitly realized in 1d\cite{Chatterjee_2025} as well as in
in 3d chiral fermion lattice models\cite{gioia2025exactchiralsymmetries31d}.
It has been emphasized in those research that the symmetries can be both non-on-site and non-compact in a few concrete lattice chiral fermion models. 

Note that all $3d$ Weyl semi-metals with the charge $U(1)$ symmetry have to have a pair or pairs of Weyl cones\cite{Burkov2018,Armitage2018}. That is fully consistent with the standard fermion doubling theorem\cite{NIELSEN1981173,NIELSEN198120}. We will be interested in a family of single Weyl fermion lattice models which evade the no-go theorem via breaking the charge $U(1)$ symmetry. As elaborated below, physically they all describe gapless superfluids or superconductors in the infrared limit and therefore naturally form an equivalence class.

A few fascinating questions about the UV completion of the IR theories at tQCPs or nodal point phases remain to be answered:

A) When constructing UV completed single chiral fermion model via tQCPs or gapless SPTs, What are the sufficient and necessary conditions to evade the fermion no-go theorem toward single Weyl fermions? Can we establish one or both of these general conditions?

B) What are the relations between the single Weyl cone infrared dynamics discussed in the context of gapless superfluids such as tQCPs or gapless topological phases\cite{Zhou1,Zhou2,Zhou3} and chiral lattice models discussed in the context of lattice chiral fermions?
Are they entirely different approaches or there is an intimate connection between two trains of thoughts?

C) If there are close connections between appearance of single Weyl cones in the context of tQCPs or magnetically polarized SPTs, and the lattice fermion models suited for the purpose of UV completion of chiral fermions, what are they?

In this article, we will make attempts to address these important questions hoping to at least improve our current views about this exciting topic.
Specifically, we will elaborate on the connections between two trains of thoughts pursued independently by different communities: emergent symmetries at topological quantum critical points or in nodal point phases in the context of gapless topological superconductors Vs lattice chiral fermions recently discussed.
And we will mainly investigate 3d lattices.

Below are a list of major findings on the questions raised above.

1) Breaking the charge $U(1)$ symmetry is one of the \textit{sufficient conditions} needed to evade the fermion no-go theorem. It can naturally lead to a single Weyl cone in topological matter. Although the original Nielsen-Ninomiya no-go theorem itself does not seem to indicate explicitly a lattice model needs to further have a non-on site symmetry, the more general arguments of t'Hooft anomalies \cite{tHooft1976,Adler1969,Bell1969} applied to emergent single Weyl cones do naturally imply the UV completion of the single-Weyl-cone low energy physics in a lattice model has to be a non-on-site symmetry. 
Studies in Ref.\cite{Fidkowski23} further suggest non-compactness.
Whether breaking the charge $U(1)$ symmetry is also one of necessary conditions remains to be further investigated. At the time of writing, we believe that seems to be likely to be true as we are not aware of other explicit 3d lattice chiral fermion models where the charge $U(1)$ symmetry is unbroken.

It is worth mentioning that in $(1+1)D$ a Chiral fermion lattice model can be constructed (via a gapless spinless fermion) with charge $U_V(1)$ symmetry unbroken\cite{Chatterjee_2025}.
In addition, the $U_A(1)$ is either non-compact or if compact, has an anomaly with $U_A(1)$ that is fundamentally non-abelian, i.e. is defined by an Onsager Algebra\cite{Chatterjee_2025,Fendley19}. (See also recent attempts in $(3+1)D$ lattices\cite{Catterall25}.)
Here we will focus exclusively on single Weyl fermion in 3d lattice models.

2) There is an intimate connection between our discussions on gapless superfluids or superconductors and emergent single Weyl cones there, and the recent 3d lattice model construction of chiral fermions. In all known studies, the charge $U(1)$ symmetry has to be broken spontaneously, either explicitly as in the discussions of gapless superfluids or in a more delicate way as in the lattice fermion construction.

All the 3d single Weyl fermion lattice models constructed so far that are known to us can be mapped, via $Spin(4)$ unitary transformations, into a $p$-wave superfluid or superconductor. In the infrared limit, they shall be isomorphic to either a gapless superconducting tQCP with the time reversal symmetry or its dual, a superconducting nodal point phase which breaks the time reversal symmetry. Therefore, these models form an equivalence class (see below) and belong to the same family. 

3) All the 3d lattice models discussed so far, if they result in single Weyl cone quantum dynamics form an equivalence class of Hamiltonians that are related by Spin(4) unitary transformations. 
The whole family of lattice Hamiltonians can be further encoded in two dual copies of $3\otimes 3$ dimensional $(1,1)$ representations of a $Spin(4)$ group. And one of the $SU(2)$ subgroups in the $Spin(4)$ group can be identified as a subgroup of emergent Lorentz $SO(3,1)$ group. In addition, the two copies of these theories are further connected by a $\sigma-\tau$ duality transformation previously introduced in Ref.\cite{Zhou1}.

The rest of our article is organized as follows. In Sect. \ref{part: realFermion}, we introduce the Nambu representation for discussions on states with charge U(1) symmetry breaking. As the Nambu representation which includes both charge $Q = \pm e$ sectors (each has spin one-half) has redundant complex fermion bands, we carry out all discussions in an equivalent real fermion representation which can be obtained by a simple unitary transformation with the degrees of freedom preserved. These real fermions form a fundamental representation of a Spin(4) group. 
Furthermore, they are always charge conjugation symmetric under the charge conjugation transformation $\cal{C}$. That is, all the real fermion theories are intrinsically charge conjugated, while for complex fermions, the charge conjugation symmetry, if relevant, needs to be further imposed extrinsically.

In this section, we will introduce the two typical paths, Path \textbf{a}) and \textbf{b}), to lattice fermion models both involve gapped fermionic topological symmetry protected states (SPTs) with the time reversal symmetry to start with. In all the discussions below, we only deal with models where the charge $U(1)$ symmetry is broken spontaneously and quantum states or phases in our discussions physically are always gapless superfluids or superconductors. There are either quantum critical points between two topologically distinct SPTs or magnetically strongly polarized SPTs. 

 Path {\bf a)}:  This first path preserves the time reversal symmetry $T$. The single Weyl cone appears when we push our gapped lattice models with protecting time-reversal symmetry or $T$-symmetry to a gapless topological quantum critical point (tQCP). The infrared effective field theory in this limit and emergent symmetries were quite extensively studied before in Refs.\cite{Zhou1,Zhou2,Zhou3}. In this article, we will exclusively focus on lattice models where one can explicitly visualize ultraviolet (UV) completions of the emergent Weyl cone dynamics and UV completed symmetry groups.

Path {\bf b)}: This second path breaks the time reversal symmetry explicitly. We peel oﬀ excessive degrees of freedom in a gapped SPT through applying T-symmetry breaking fields to the SPT which naturally result in a pair of gapless real fermion nodal points. A lattice model was previously constructed in ref.\cite{Kapoor} to explicitly verify the existence of such a phase where only a pair of real fermion nodal points appear.
That study was to verify that only single Weyl cone emerges in the low energy sector near a pair nodal points of real fermion bands\cite{Yang_Zhou} in the whole 3-torus
$\mathbb{T}^3$ of the crystal momenta.

In this article, when examining these main paths, we will mainly explore the UV completion of the emergent symmetries associated with the single Weyl cone physics in the infrared limit. We carry out all our discussions via lattice models where the whole momentum space of 3-torus $\mathbb{T}^3$ can be tracked explicitly.

In Sect.\ref{part: NN-no-go-theorem}, we revisit the fermion no-go theorem on band crossings, but carry out our examination on real fermions rather than the usual complex fermions so to apply directly to gapless superconductors or superfluids. Instead of using topological homotopy analyses and investigating the topological curvature flux emitted from crossings, here we employ an alternative differential-geometry-based approach. 

We introduce a six dimensional manifold that consists of two three-dimensional oriented submanifolds. One of such three-dimensional submanifolds is spanned by the Hamiltonian fibers, $H(\bf p)$ which are momentum ${\bf p}$ dependent, and the other three dimensional submanifold is $\mathbb{T}^3$ for the three dimensional lattice momenta as a base space. Whenever the two oriented submanifolds intersect, we show that two real fermion bands cross each with a specific handedness. We then identify that the numbers of Weyl cones, left or right, can be directly mapped into the numbers of intersections between two oriented three-dimensional manifolds. 

We show that for real fermions, the intersections always appear in pairs at $\pm {\bf p}_i$,$i = 1,2,3...,M$. The even-integer numbers of interactions follows an intersection theorem in differential geometry.
The momentum space structure of these intersections, i.e. intersections always appear in pairs at $\pm {\bf p}$,
is a result of the intrinsic charge conjugation symmetry of real fermions present in any charge $U(1)$ symmetry breaking state. 

These intersections therefore indicate $2M$ real fermion band crossings with $M = 1,2,3,.....$. The case of $M = 1$ has one pair of real fermion band crossings at $-\mathbf{p}$ and $\mathbf{p}$ respectively. Because of the charge conjugation symmetry, it can be exactly mapped into an emergent Weyl fermion cone via a straightforward explicit reconstruction.
In addition, $M=3,5,..$ lead to odd numbers of Weyl fermions.

In the real fermion representation, gapless 3D Dirac fermion cones with both left and right Weyl fermion cones would appear as the limit where $M=2,4,6,...$ i.e. there are even numbers of pairs of band crossings. This is the manifestation of the standard no-go theorem of complex fermions with the usual charge $U(1)$ symmetry\cite{NIELSEN198120,NIELSEN1981173}, but in the real fermion representation.

In Sect.\ref{part: Model_analysis}, we analyze a few lattice chiral fermion models through the lense of the real fermion framework introduced in the previous sections. In this part, we focus on a model along Path \textbf{a}) 
which is time reversal symmetric and and Path \textbf{b}) which breaks the time reversal symmetry explicitly.
We also compare the non-on-site symmetries in different lattice models.

Model I represents a lattice model of a generic tQCP with protecting time-reversal $T$-symmetry that involved a minimum change of topologies with $\delta N^f_w=2$. Here $\delta N_w$ is the change of topological invariants across a tQCP and the superscript \textit{f} further refers to the fundamental value. It is a tQCP in the well-known DIII class of topological superconductors with all the UV completed details in 3-torus $\mathbb{T}^3$.

Model II represents a lattice model of a nodal point phase induced by strong magnetic field acting on a fully gapped topological superconductor in a DIII class with protecting T-reversal symmetry (but again with charge $U(1)$ symmetry spontaneously broken).

In Sect.\ref{part: Model_analysisII}, we will focus on three lattice models along Path \textbf{c}) which is a hybrid of Path \textbf{a}) and Path \textbf{b}) that involves both tQCPs and time reversal symmetry breaking fields.
We start with the model III, a tQCP with a minimum change of topologies $\delta N_w=2$ but further subject to a time reversal symmetry breaking magnetic field.
The magnetic field lifts the degeneracy of Kramer doublets and results in two real fermion band crossing points at $\pm {\bf p}_0$.

Model IV is a multicritical tQCP that leads to a change of topologies twice the fundamental value, i.e. $\delta N_w=4$. 
Model V is the closely related to the one introduced in ref.\cite{gioia2025exactchiralsymmetries31d} where $\delta N_w=8$, quadrupling of the fundamental unit of $\delta N^f_w=2$.

In both Sect.\ref{part: Model_analysis}, \ref{part: Model_analysisII}, we also illustrate the UV completion of the emergent symmetries in the infrared limit and construct explicitly the symmetry charge operators in torus-three $\mathbb{T}^3$. We further briefly discuss different phases 
appearing in these lattices models and identify the phase boundaries for the phases with an emergent single Weyl fermion.

In Sect.\ref{part: family},
We then further illustrate how all different models discussed in the previous two sections and known to us can be organized into a distinct linear representation of $Spin(4)$. The infrared (IR) limit of all the lattice models discussed so far, if they lead to single Weyl cone dynamics, shall always be a part of a small family of Hamiltonians that can be related to each other via Spin(4) transformations. Both the lattice Hamiltonians along with their infrared limits form an equivalence class.

Here specifically, we further show this entire family of lattice Hamiltonians can be encoded in two copies of $3\otimes3$ dimensional $(1,1)$ representations of Spin(4) group where one of the $SU(2)$ subgroups can be identified as a subgroup of emergent Lorentz SO(3,1) group. These Hamiltonians span two dual linear spaces of $Spin(4)$ group. In addition, these two copies of the representations or linear spaces are further connected by a $\sigma-\tau$ duality transformation previously introduced in Ref.\cite{Zhou1}.

In Sect. \ref{part: Charge}, we discuss the construction of the non-compact symmetry group along different paths, path ${\bf a})$ and path ${\bf b}, {\bf c}$. The non-compact nature of symmetry charges was previously emphasized in Ref.\cite{Friedan:1982nk,Fidkowski23,Chatterjee_2025,gioia2025exactchiralsymmetries31d}.
We illustrate that along path \textbf{a}) where the gapless tQCP states are time-reversal invariant, the linear space spanned by the conserved symmetry charge operators is a six-dimensional manifold. Along path \textbf{b}) and path \textbf{c}) where the magnetic peeling is applied and the time reversal symmetry is broken, the dimension of the linear space of the conserved charge operators is two-dimensional. We discuss the implications of these results.

In Sect.\ref{part: Conclusion}, we conclude our studies and discuss a few
open questions on this topic.

\section{The systematic real fermion approach}
\label{part: realFermion}

In this section we lay the technical foundation of our approach by reformulating lattice fermion models in terms of real (Majorana) fermions. This reformulation is essential for treating systems with broken U(1) charge symmetry—a necessary ingredient in our construction. We show the equivalence of the BdG and real representations, and then use this formalism to introduce two distinct paths (time-reversal symmetric and time-reversal breaking) toward constructing single Weyl cones on the lattice.

\tdplotsetmaincoords{65}{115} 

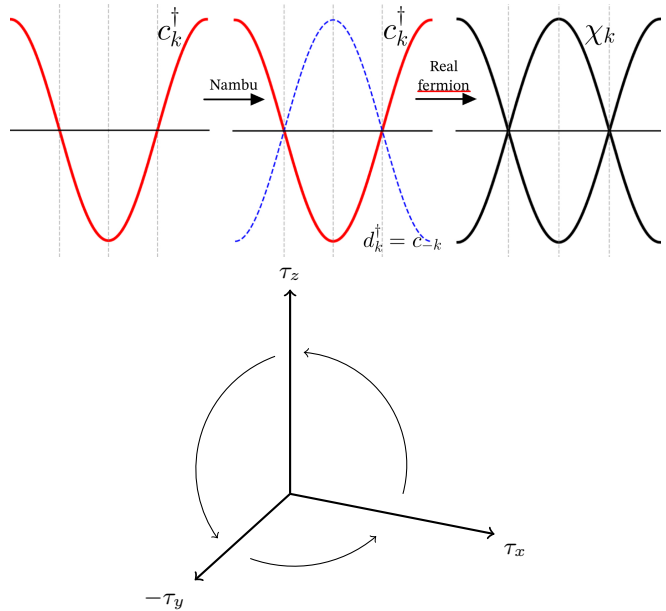
\begin{figure}[h]
    \centering
    \includegraphics[width=1.\linewidth]{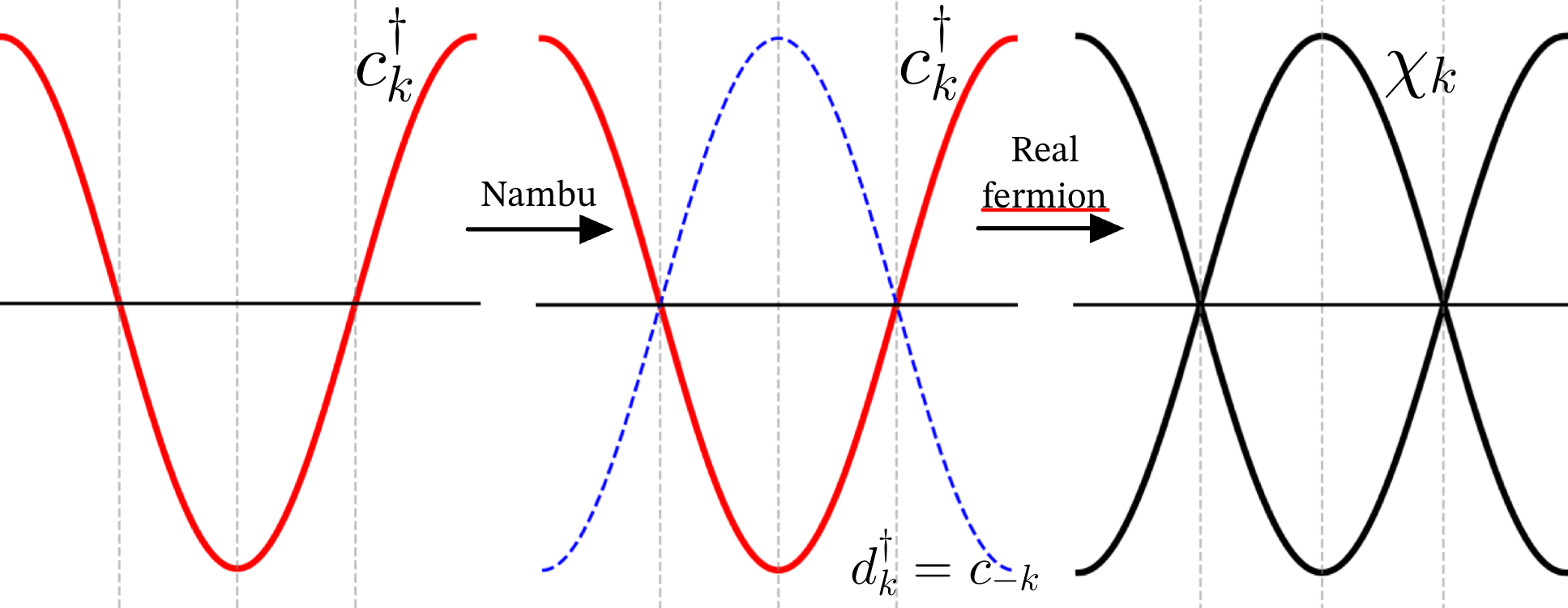}
    \begin{tikzpicture}[tdplot_main_coords, line join=round]

  \draw[->, thick] (0,0,0) -- (3,0,0) node[anchor=north east]{$-\tau_{y}$};
  \draw[->, thick] (0,0,0) -- (0,3,0) node[anchor=north west]{$\tau_{x}$};
  \draw[->, thick] (0,0,0) -- (0,0,3) node[anchor=south]{$\tau_{z}$};
  \draw[->,  black, tdplot_screen_coords] 
  (1.5,0) arc[start angle=-15, end angle=85, radius=1.5cm] ;
  \draw[->,  black, tdplot_screen_coords] 
  (-0.2,1.83) arc[start angle=110, end angle=215, radius=1.6cm] ;
  \draw[->, black, tdplot_screen_coords] 
  (-0.53,-0.86) arc[start angle=-110, end angle=-50, radius=1.7cm] ;

\end{tikzpicture}
    \caption{Top row from the left to right: shown on the left is one complex fermion degree of freedom at each k associated with the creation operator $c_{\bf k}^\dagger$. In the middle figure one further redefines the fermion creation operator as an annihilation operator for a hole-like fermion so that $c_{\bf k}^\dagger = d_{-{\bf k}}$. In the right figure, one rotates to the real fermion basis. The red line is for charge $+$e and the blue dotted line is for charge $-$e.
    The bottom figure represents the transformation from the Nambu space to a real fermion representation.}
    \label{fig:Nambu}
\end{figure}

\subsection{Real fermion representation}

One way to avoid the Nielsen-Ninomiya no-go theorem \ref{thm: NN} is to break the charge U(1) symmetry, the exact conservation of the U(1) charge being the main assumption of the theorem, and the natural first step towards doing so is to consider real fermions. 

Real fermions are more fundamental than complex fermions: any n-component complex fermions can be constructed out of 2n-component real fermions by imposing the U(1) symmetry. Consider:
\begin{equation}
    \chi = (\chi_1,...,\chi_{2n})^T, \chi_i^\dagger=\chi_i, i=1,..n,n+1,....2n
\end{equation} as real fermions satisfying the following anti-commuting relations 

\begin{equation}
    \{\chi_i , \chi_j\} = \delta_{i,j} ,\quad i,j = 1,...,2n.
\end{equation}

These real fermions then can also describe the $n$ complex fermions by the following constructions
\begin{equation}
\psi = \frac{1}{\sqrt{2}}(\chi_1+i\chi_{n+1},...,\chi_{n}+i\chi_{2n})^T.
\end{equation}
One can easily verify the anti-commuting relations between these fermion fields.

These complex fermions then can carry the standard conserved $U(1)$ charges if and only if the Hamiltonian is further invariant under the following $SO(2) \cong U(1)$ transformation,

\begin{equation}
    R(\theta) = \begin{pmatrix}
    \cos(\theta/2)I_n&-\sin(\theta/2)I_n\\
    \sin(\theta/2)I_n&\cos(\theta/2)I_n
\end{pmatrix},
\end{equation} 
where $\theta \in [0,4\pi]$, and $I_n$ is an $n\times n$ identity matrix. 

\subsection{Nambu Representation and Real Fermions}

For charge $U(1)$ symmetry breaking states, we typically use the Nambu representation. Below we illustrate that Nambu representation can be easily rotated into a real fermion representation, a result that had been utilized in many previous studies of topological superfluids and superconductors\cite{Yang_Zhou, Zhou1}. 

Consider a general case where we have a 2 band Hamiltonian for complex fermion with a U(1)-symmetry-breaking term (a factor 2 was added to simplify notations): 

\begin{equation}
    \mathcal{H}_0 =\int_{BZ} d\mathbf{k}\left[c_k^\dagger H_0(\mathbf{k})c_k + \frac{1}{2}(c_{-k}^T \Delta(\mathbf{k})c_{k} + H.c.)\right]
\end{equation} 
Where $c_k= (
    c_{k \uparrow} , c_{k \downarrow})^T$
represents spin-1/2 complex fermions. $H_0(\mathbf{k})$ is a standard free particle Hamiltonian.
And $\Delta(\mathbf{k})$ are $2\times2$ complex matrices; the matrix elements $\Delta_{s,s'}({\bf p})$ are subject to the symmetry constraints of 
$\Delta_{s,s'}({\bf p}) =- \Delta_{s',s}(-{\bf p})$, with $s,s'=\uparrow,\downarrow$ as spin indices.

We can rewrite it in the following way (we omit the spin indices):

\begin{widetext}
\begin{equation}
    \mathcal{H}_0 
    =  \frac{1}{2}\int_{BZ} d\mathbf{k} [\,c_k^\dagger H_0(\mathbf{k})c_k - c_{-k}^TH_0^*(-\mathbf{k})(c_{-k}^\dagger)^T + c_{-k}^T \Delta(\mathbf{k})c_{k}-c_{k}^\dagger \Delta^*(-\mathbf{k})(c_{-k}^{\dagger})^T\,]
\end{equation}
\end{widetext}
with $\Delta^\dagger(\mathbf{k}) = -\Delta^*(-\mathbf{k})$. So we have an alternative way of writing the Hamiltonian, using the basis $(c_{k\uparrow},c_{k\downarrow},c_{-k\uparrow}^\dagger,c_{-k\downarrow}^\dagger)$, which is known as the Bogoliubov–de Gennes formalism:

\begin{equation}
\label{eq:BdG_H}
    H_{\text{BdG}}(\mathbf{k}) = 
 \frac{1}{2}\begin{pmatrix}
    H_0(\mathbf{k}) &  -\Delta^*(-\mathbf{k})\\
    \Delta(\mathbf{k}) & -H_0^*(\mathbf{-k})
\end{pmatrix}
\end{equation}

As is obvious from the formulation (Eq.\ref{eq:BdG_H}), $H_{\text{BdG}}$ has one half of the number of degrees of freedom as a $4$-component complex fermion Hamiltonian. Thus it is natural to reformulate it in terms of 4-band real (majorana) fermions as illustrated in Fig.\ref{fig:Nambu}. 

The real fermion operators are defined as follows:
\begin{equation}
\begin{aligned}
    \chi_{+,s}(x) &= \frac{1}{\sqrt{2}} (c_{s}(x) + c^{\dagger}_{s}(x)),\\
    \chi_{-,s}(x) &= \frac{1}{i\sqrt{2}} (c_{s}(x) - c^{\dagger}_{s}(x))
\end{aligned}
\end{equation}
$s = \uparrow,\downarrow$ representing the spin indices.

The 4-band real fermions are defined as
\begin{equation}
    \chi(x) = \begin{pmatrix} \chi_{+,\uparrow}(x) & \chi_{+,\downarrow}(x)&\chi_{-,\uparrow}(x)&\chi_{-,\downarrow}(x) \end{pmatrix}^T
\end{equation}
so that $\chi^{\dagger}(x) = \chi^T (x)$.
So the change of basis writes:

\begin{equation}
\label{eq: trans_majorana}
    \begin{pmatrix} c_\uparrow(x)\\  c_\downarrow(x) \\ c_\uparrow^{\dagger}(x) \\ c_\downarrow^{\dagger}(x) \end{pmatrix} =  \frac{1}{\sqrt{2}} \begin{pmatrix} I_2 & iI_2 \\ I_2 & -iI_2 \end{pmatrix}  \chi(x)
\end{equation}
where $I_2$ acts on the spin subspace. 

This transformation is equivalent to the following $SU(2)$ unitary transformation in the Nambu space (up to a global $U(1)$ phase factor):

\begin{equation}
\label{eq:R_trans}
    R = \frac{I}{2} - \frac{i}{2}(\tau^x -\tau^y +\tau^z)
\end{equation}
which acts as:

\begin{equation}
    \tau^x \rightarrow \tau^z, \quad \tau^y \rightarrow -\tau^x, \quad \tau^z \rightarrow -\tau^y
\end{equation}

In what follows, we will exclusively work in the real fermion basis.
After this unitary transformation, (Eq.\ref{eq:BdG_H}) can be rewritten as:

\begin{widetext}
    \begin{equation}
    \label{eq: nambu after R}
        H_{\text{Rotated}}(\mathbf{k}) = \frac{1}{2}\left[\left(\mathbb{I}_2\otimes\frac{H_0(\mathbf{k})-H_0^*(-\mathbf{k})}{2}  -  \tau^y\otimes\frac{H_0(\mathbf{k})+H_0^*(-\mathbf{k})}{2} \right)+ \tau^z\otimes\frac{\Delta(\mathbf{k})-\Delta^{*}(-\mathbf{k})}{2} - \tau^x\otimes\frac{\Delta(\mathbf{k})+\Delta^*(-\mathbf{k})}{2i}\right]
    \end{equation}
\end{widetext}

Let us now show that every real fermion Hamiltonian can be written in the form (Eq.\ref{eq: nambu after R}).

Let $H_{\text{Real}}$ be a real fermion Hamiltonian. The real fermions satisfy the charge conjugation symmetry: 
\begin{align}
    \chi^\dagger(x) &= \chi^T(x),\\
    \label{eq_realF}
    \Rightarrow \chi^\dagger(-\mathbf{k}) &= \chi^T(\mathbf{k})
\end{align}
which implies that $H_{\text{Real}}$ must follow:

\begin{equation}
    H_{\text{Real}}^*(x) + H_{\text{Real}}(x) = 0,
\end{equation}
\begin{equation}
\label{eq:H_real}
    \Rightarrow H_{\text{Real}}(\mathbf{k}) = -H_{\text{Real}}(-\mathbf{k})^*
\end{equation}
For the purpose of our demonstration, we decompose $H$ in the following way:
\begin{equation}
    H_{\text{Real}}(\mathbf{k})=\frac{1}{2}\left[\mathbb{I}_2\otimes A  -  \tau^y\otimes B+ \tau^z\otimes C +i \tau^x\otimes D\right]
\end{equation}
with $A$ and $C$ being $2\times2$ hermitian matrices and $B$ and $D$ being $2\times2$ anti-hermitian matrices.

The real Hamiltonian constraint (Eq.\ref{eq:H_real}) then translates as:
\begin{subequations}
\begin{align}
    A(\mathbf{k})&=-A(-\mathbf{k})^*\\
    C(\mathbf{k})&=-C(-\mathbf{k})^*\\
    B(\mathbf{k})&=B(-\mathbf{k})^*\\
    D(\mathbf{k})&=D(-\mathbf{k})^*
\end{align}
\end{subequations}
Thus if we set $H_0 = A+B$ and $\Delta = C+D$, then:
\begin{subequations}
\begin{align}
    A(\mathbf{k})&=\frac{H_0(\mathbf{k})-H_0^*(-\mathbf{k})}{2} \\
    B(\mathbf{k})&=\frac{H_0(\mathbf{k})+H_0^*(-\mathbf{k})}{2} \\
    C(\mathbf{k})&=\frac{\Delta(\mathbf{k})-\Delta^*(-\mathbf{k})}{2} \\
D(\mathbf{k})&=\frac{\Delta(\mathbf{k})+\Delta^*(-\mathbf{k})}{2}
\end{align}
\end{subequations}
gives back the form (Eq.\ref{eq: nambu after R}) for $H_{\text{Real}}$. 

Therefore we have shown that the Nambu representation is equivalent to the Real representation, further justifying the use of the real representation in our analyses.

Furthermore, real fermions have a simpler charge conjugation symmetry relation compared to the BdG formalism;
in the BdG formalism, particle hole symmetry writes:

\begin{equation}
    \tau^x H_{\text{BdG}}(\mathbf{k})\tau^x = - H_{\text{BdG}}^*(-\mathbf{k})
\end{equation}
whereas in the real fermion basis it comes down to just (Eq.\ref{eq:H_real}), as one can verify that $R\tau^xR^T = iI$.

\subsection{Lattice models of real fermions}

\label{part: 2_general_single_weyl}

We will study in details four concrete 3D lattice models in section \ref{part: Model_analysis} that can be thought of as some particular cases of the general model that we construct below. Starting with a gapped fermionic topological symmetry protected state (SPT) with the T-symmetry, i.e. a {\em DIII} class topological superconductor, we can pursue along the two distinct paths {\bf a)} and {\bf b)} towards the construction of lattice chiral fermions.

Let us emphasize again here that physically all the 3D lattice model under our considerations lead to gapless superconductors or superfluids where the conventional charge $U(1)$ symmetry is broken spontaneously. They all lead to an emergent single Weyl cone physics in the infrared limit. However, our main interest in this article is about the UV completion of infrared Weyl 
fermions in the momentum space of $\mathbb{T}^3$. 

\begin{figure}[h]
    \centering
    \includegraphics[width=0.9\linewidth]{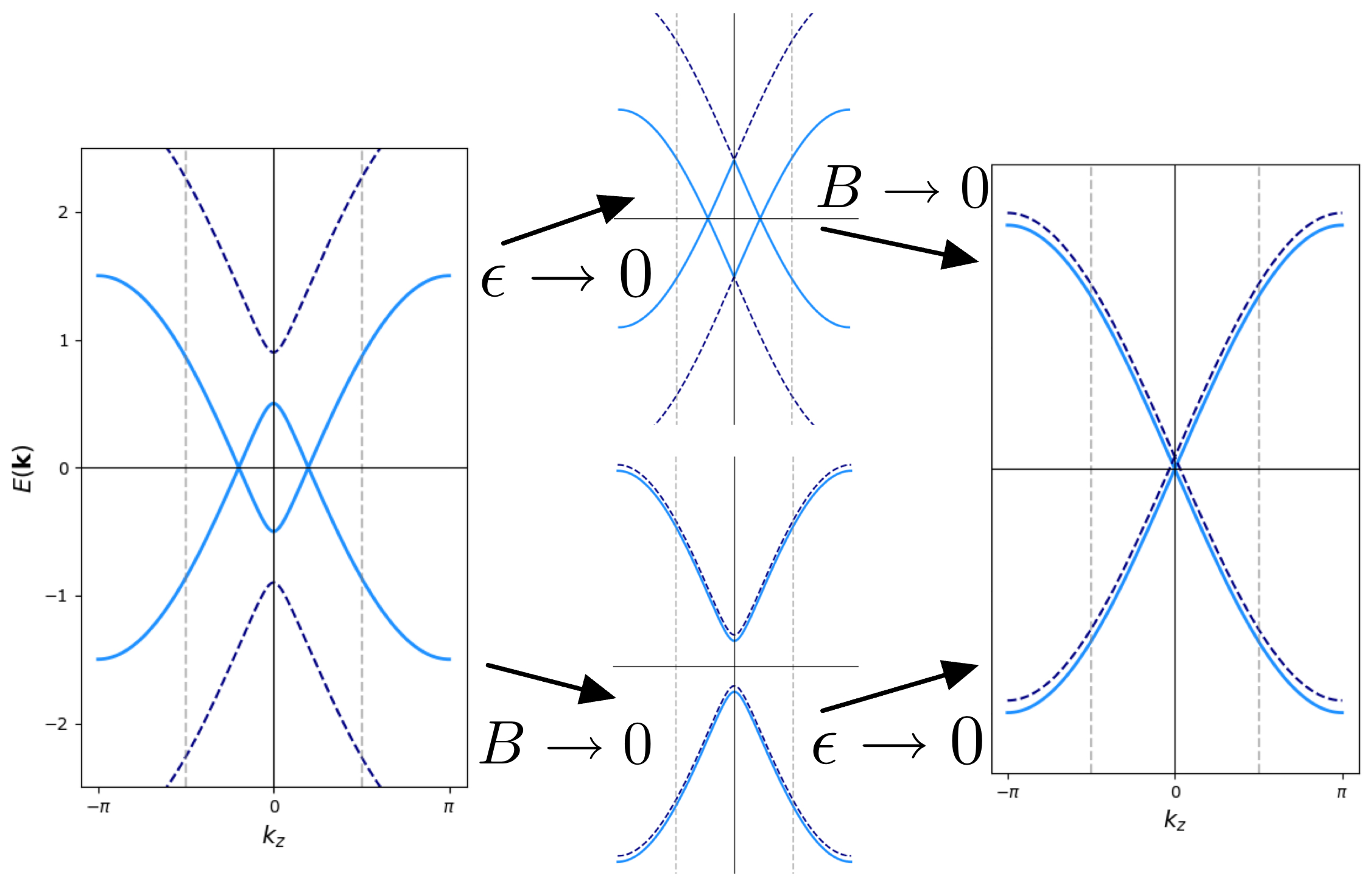}
    \caption{The $T$-invariant tQCP (Right) as a extreme case of the nodal phase (Left). A $T$-invariant tQCP can be seen as a fine tuned phase of the Nodal phase, either by tuning the coupling parameter $\epsilon$ to 0 and then the T-breaking magnetic field (up path), or by first tuning B to 0, coming back to the gapped lattice model with T-symmetry, and then pushing it to a gapless phase.}

\label{fig:Nodal_vs_tQCP2}
\end{figure}

Path {\bf a)}: We push the gapped state to be a quantum critical one while preserving the time reversal symmetry. Especially, we close the gap of our SPT at one of the 8 corners of the Brillouin zone say at at $\mathbf{k} = 0$, while maintaining a finite gap everywhere else. The state obtained in this way is a Time reversal invariant three dimensional tQCP.

Path {\bf b)}: we apply a $T$-symmetry breaking field to the fully gapped three dimensional SPT which naturally results in a pair of gapless nodal points. This pair can be further applied to reconstruct a single Weyl cone (Fig\ref{fig:Nodal_vs_tQCP2}).

\subsubsection{Path {\bf a)}}

Following our discussion on real fermions, we will work with a general 4-band Hamiltonian in the real fermion representation:

\begin{equation}
    \mathcal{H}=\sum_{\bf k} \sum _i \chi^T ({\bf k}) \,A_i(\mathbf{k})\Gamma_{i}\, \chi(-{\bf k}), \quad \Gamma_i=\Gamma^\dagger_i, \Gamma_i^2=I
\end{equation}
where $\Gamma_i$ are hermitian matrices.

Contrary to Clifford algebras where only anticommuting matrices are involved, here $\Gamma_{i}$ and $\Gamma_{j}$ can either commute or anti-commute with each other to account for our discussion of path b) where we need commuting matrices.

However, because $\Gamma_i$ are coupled to real fermion fields, there are stringent constraints due to the real fermion anti-commuting algebras unlike in the complex fermion representation. That is, using condition (Eq.\ref{eq:H_real}), we write
\begin{equation}
    A_i(\mathbf{k})\Gamma_i = A_i^Q(\mathbf{k})\Gamma_i^S+A_i^A(\mathbf{k})\Gamma_i^A
\end{equation}
Where $\Gamma^S$ is the symmetric part of $\Gamma$ and $\Gamma^A$ the antisymmetric part.
Then the constraint writes:

\begin{align}
    A_i^Q(\mathbf{k})&=-A_i^Q(\mathbf{-k})\\
    A_i^A(\mathbf{k})&=A_i^A(\mathbf{-k})
\end{align}

Let us choose a particular realization of the coefficients $A_i({\bf k})$:

\begin{eqnarray}
    H(\mathbf{k}) = \sin(k_x)\Gamma_1 + \sin(k_y)\Gamma_2 + \sin(k_z)\Gamma_3+M({\bf k})\Gamma_4; \nonumber \\
    \label{eq: tQCP0}
\end{eqnarray}

$\Gamma_{i}$, $i=1,2,3,4$ follow the standard Clifford algebra. And they transform under the time reversal transformation $T$ as

\begin{eqnarray}
 && \{ \Gamma_i, \Gamma_j \}=2\delta_{ij}, \mathcal{T}^2=-1;\nonumber \\
&& \mathcal{T}^{-1} \Gamma_{i} \mathcal{T}=-\Gamma_i, i=1,2,3;
\mathcal{T}^{-1} \Gamma_{4} \mathcal{T}=\Gamma_4.
\label{eq:CAlgebra}
\end{eqnarray}

The coefficient $M({\bf k})$ is a T-symmetry preserving coupling, so the bands are Kramer degenerate everywhere. It can be seen as an effective model of a p-wave superconductor in the limit of strong coupling. As long as this coefficient is non-zero at all of the 8 points $k_x, k_y,k_z = 0,\pi$, the superconducting SPT state remains gapped. We can choose:

\begin{equation}
    M({\bf k}) = \mu-\sum \cos(k_i),\mu=3+\epsilon
    \label{eq: mass}
\end{equation}
where $\mu$ is a mass parameter. In DIII class superconducting SPTs, $\mu$ can be related to chemical potentials.

From algebraic considerations, we obtain the following spectrum:

\begin{equation}
E_\pm^2 = \sin(k_x)^2+\sin(k_y)^2+\sin(k_z)^2+(3+\epsilon-\sum \cos(k_i))^2
\end{equation}

\begin{figure}[h]
    \centering
    \includegraphics[width=0.9\linewidth]{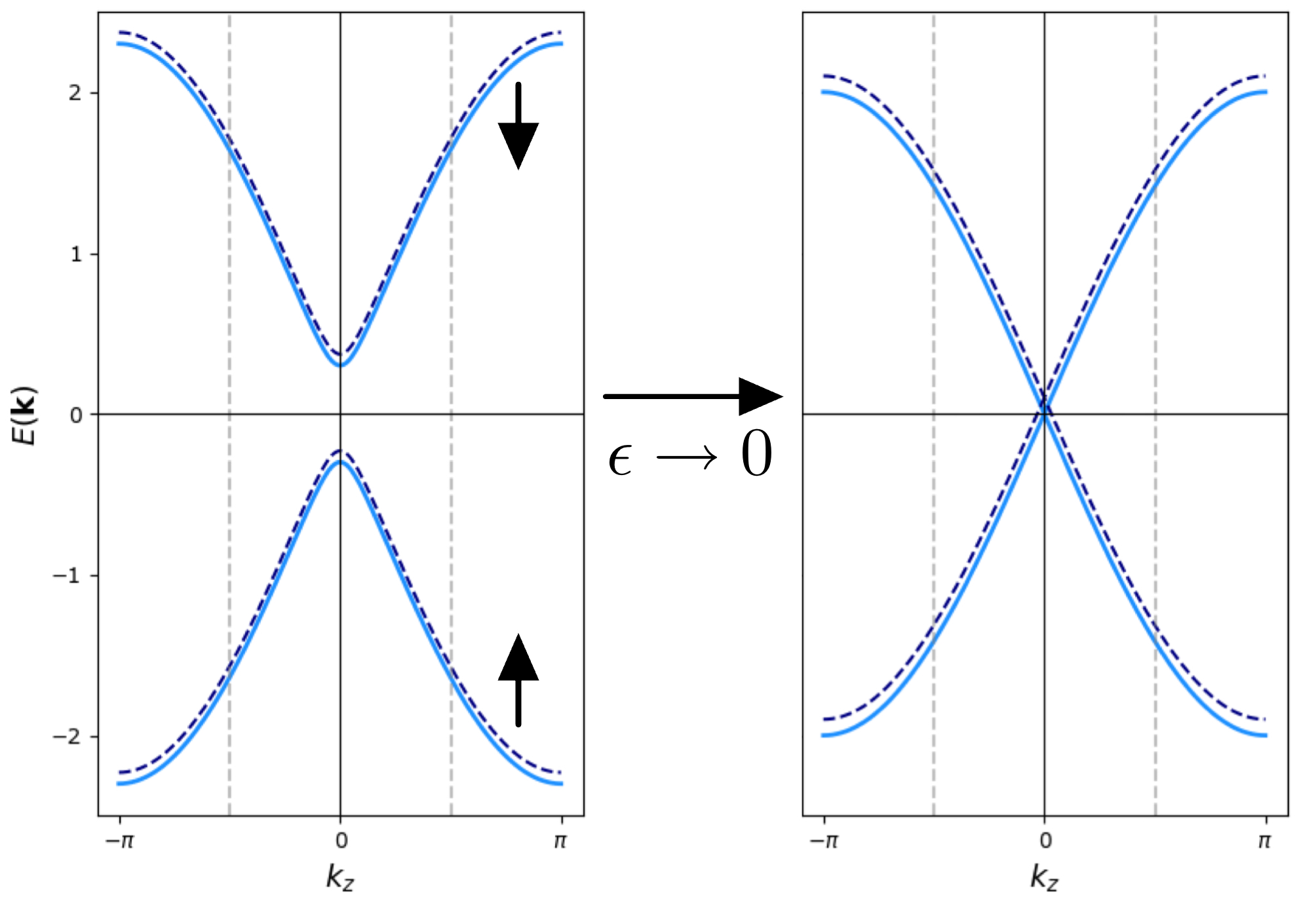}
    \caption{A T-symmetric tQCP in SPTs when the mass parameter is tuned to be zero, i.e. $\epsilon = 0$ (See Eq.\ref{eq: tQCP0},\ref{eq: mass}).}

\label{fig:tQCP}
\end{figure}

In the lattice model Eq.\ref{eq: tQCP0},\ref{eq: mass},
$\epsilon>0$ corresponds to a trivial state, and $\epsilon<0$ to a topological state. The gap closes for $\epsilon = 0$ at $\mathbf{k} = 0$ only, therefore realizing a tQCP that has the time-reversal symmetry (FIG. \ref{fig:tQCP}).

\subsubsection{Path {\bf b)}}

We start with the previous model in a gapped phase, $\epsilon\neq 0$, and we add a T-breaking field (e.g. a magnetic field along the z-direction). We replaced $3+\epsilon$ by $\mu$ as $\epsilon$ is no longer supposed to go to zero. For the purpose of the current study, we consider the magnetic field $B$ to be k-independent:

\begin{equation}
\label{eq: NodalH}
\begin{aligned}
     H(\mathbf{k}) =  &\sin(k_x)\Gamma_1 + \sin(k_y)\Gamma_2 + \sin(k_z)\Gamma_3 \\&+M({\bf k})\Gamma_4 + B \Gamma_5
\end{aligned}
\end{equation}
where the mass $M({\bf k})$ has been introduced in Eq.\ref{eq: mass}.

And following the algebras in Eq.\ref{eq:CAlgebra}. one can further set 

\begin{eqnarray}
     && \{\Gamma_5, \Gamma_1 \} =\{ \Gamma_5, \Gamma_2 \}=0; \nonumber \\
      && [ \Gamma_5, \Gamma_3 ] =[ \Gamma_5, \Gamma_4 ]=0.
\end{eqnarray}
That is $\Gamma_5$
commutes with $\Gamma_3$ and $\Gamma_4$ and anti-commutes with $\Gamma_1$ and $\Gamma_2$. Without losing generality, one can choose
an antisymmetric hermitian operator $\Gamma_5$ of the following form, 

\begin{equation}
    \Gamma_5=i \Gamma_1\Gamma_2, \Gamma_5^T=-\Gamma_5, \mathcal{T}^{-1} \Gamma_5 \mathcal{T}=-\Gamma_5.
\end{equation}
Note it does break the time reversal $T$-symmetry as desired.

\begin{figure}[h]
    \centering
    \includegraphics[width=1.\linewidth]{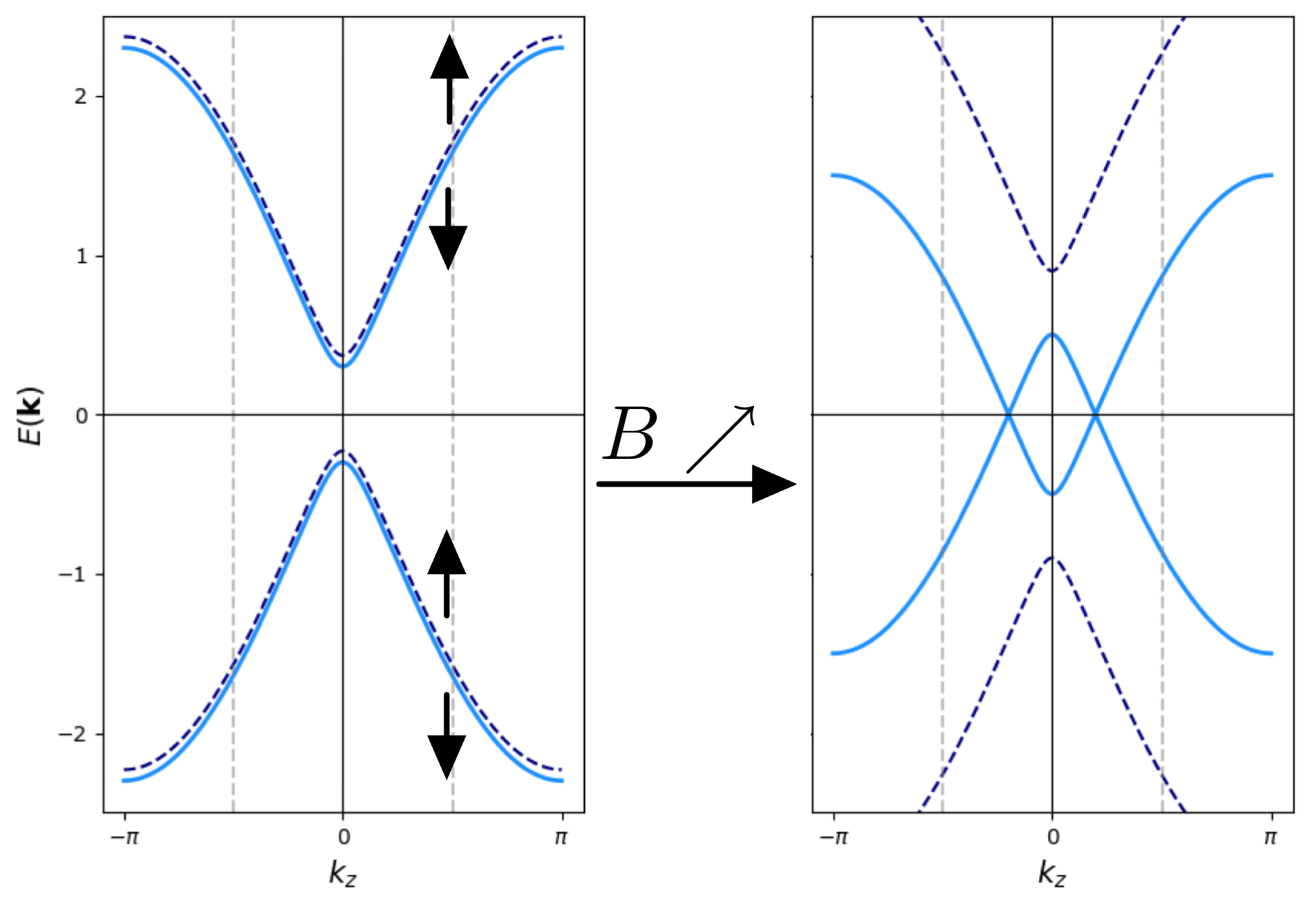}
    \caption{(Nodal phase) Spectrum for $\mu = 2.5$ and, left $B=0$, the bands are 2 times degenerate; right $B=4$, the magnetic field parameter $B$ lifts the degeneracy and for a range of values gives only two crossings. This allows the peeling off of the excessive degrees of freedom (dotted bands) in this real fermion formalism, a crucial step towards the single Weyl fermion.}
    \label{fig:Band_structure_real_weyl}
\end{figure}

We have the following spectrum (FIG. \ref{fig:Band_structure_real_weyl}):

\begin{equation}
    \begin{aligned}
    E_\pm^2 = &\left(B\pm\sqrt{(\mu - \sum \cos(k_i))^2 + \sin (k_z)^2}\right)^2 \\
    &+ \sin(k_x)^2+\sin(k_y)^2
\end{aligned}
\end{equation}

We want to tune the parameters $\mu$ and $B$ so that there are only two degeneracy points at zero energy. The range of values where this is true is shown in FIG.\ref{fig:phase diagram}.
\begin{figure}[h]
    \centering
    \includegraphics[width=0.9\linewidth]{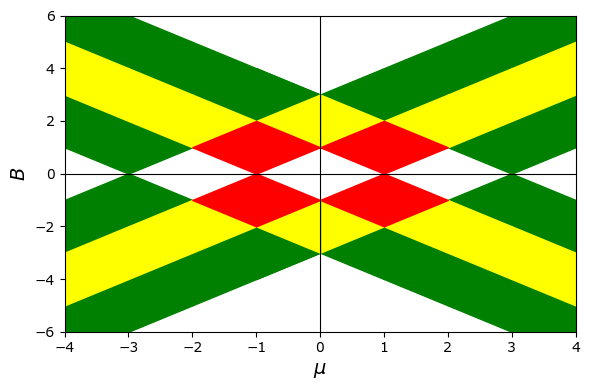}
    \caption{Phase diagram for the number of Weyl cones in model Eq.\ref{eq: NodalH}. In the white region, there are no Weyl cones, in the green region, there is only one single Weyl cone, in the yellow region, there are two Weyl cones and in the red region, there are three Weyl cones.}

\label{fig:phase diagram}
\end{figure}

So what we finally get are two degeneracy points at opposite momenta related by charge conjugation symmetry. This phase is robust against deformations of the Hamiltonian as crossing points are separated in the momentum space, and it is a typical nodal point phase.

In the next section, we will show that an emergent single Weyl fermion is a generic feature of gapless real fermion lattice models
and hence can naturally appear in gapless superconductors or superfluids. Readers who are not interested in this general proof
can skip Sect.\ref{part: NN-no-go-theorem} and proceed directly to Sect.\ref{part: Model_analysis}, \ref{part: Model_analysisII} where concrete models are presented and discussed.

\section{Naturalness of single Weyl fermion in gapless superconductors or superfluids}

\label{part: NN-no-go-theorem}
This section provides an alternative proof of the Nielsen-Ninomiya theorem using tools from differential geometry. Rather than relying on topological charges and Berry curvature, we reinterpret band crossings as intersections of submanifolds in a higher-dimensional space. We then apply this geometric perspective to real fermions, showing how particle-hole symmetry naturally enforces the pairing of band crossings, yet allows for the reconstruction of a single Weyl fermion through specific projections.

Physically, this suggests naturalness of single Weyl fermion in gapless charge $U(1)$ symmetry breaking states such as gapless superconductors or superfluids. They shall form an equivalence class.
Later in Sect.\ref{part: family}, we provide an explicit structure of the equivalence class of Hamiltonians.

\subsection{alternative proof of Nielsen and Ninomiya No-Go theorem}

\label{proof_nogo1}

Here we show an alternative proof of the celebrated Nielsen-Ninomyia no-go theorem \cite{NIELSEN1981173,NIELSEN198120}.
The proof presents a fundamentally different approach to the problem than the original one \cite{NIELSEN198120} and the other ones that we could find in the literature \cite{Friedan:1982nk},\cite{Kiritsis:1986re}. The intuitive topological proof given by Nielsen and Ninomiya \cite{NIELSEN1981173} uses intersection theory as well, but with a projective line and a 3-torus embedded in a four dimensional space. Our approach is more systematic and has the advantage of showing the handedness of the Weyl points directly from the formalism. We believe it is potentially generalizable to other dimensions. Despite requiring some knowledge on differential geometry, this proof is actually quite natural. We present the idea of the proof here (with a $2\times2$ Hamiltonian) and the complete mathematical proof, as well as some refinements, will be reported in the Appendix \ref{A_proof no go}.

The idea, based on the theory of intersections in differential geometry \cite{1974Dt/b}, provides an elegant argument of why the no-go theorem exists only in dimension 3.

\begin{theorem}[Nielsen \& Ninomiya]
\label{thm: NN}
    In a lattice theory with local interaction Hamiltonian ($H(x-y) \rightarrow 0$ fast enough when $x-y \rightarrow\infty$) that is invariant under lattice translation, the number of right-handed and left-handed Weyl fermions is equal, granted the following assumptions on the charge $Q$ are satisfied:
    \begin{itemize}
        \item $Q$ is exactly conserved: $[Q,H] = 0$
        \item $Q$ is locally defined
        \item $Q$ is quantized
        \item $Q$ is bilinear in the fermion field.
    \end{itemize}
\end{theorem}

\begin{proof}
    In this paragraph we restrict ourselves to a $2\times2$ Hamiltonian. We let $\text{Herm}(2)$ be the set of $2\times2$ Hermitian matrices. A 2-band Hamiltonian around degeneracy points can be written as
\[
H(\mathbf{p}) = \epsilon + (\mathbf{p} - \mathbf{p}_{\text{deg}}) \mathbf{b} + (\mathbf{p} - \mathbf{p}_{\text{deg}})_{\kappa} \mathbf{V}^{\kappa}_{\alpha} \sigma^{\alpha} + O((\mathbf{p} - \mathbf{p}_{\text{deg}})^{2})
\]

The handedness of the degeneracy point is determined only by the sign of the determinant of $V$, if it is positive it is right-handed, if it is negative it is left-handed. We can see that by subtracting the identity part from H and by redefining the momentum as is done in \cite{NIELSEN198120}:

\begin{align*}
    H(\mathbf{p}) &\leftarrow H(\mathbf{p}) - \epsilon + (\mathbf{p} - \mathbf{p}_{\text{deg}}) \mathbf{b}\\
    \mathbf{P_{\alpha}} &= \text{sign(det(}\mathbf{V}))\cdot(\mathbf{p} - \mathbf{p}_{\text{deg}})_{\kappa} \mathbf{V}^{\kappa}_{\alpha}
\end{align*}

So that we have:

$$H(\mathbf{p}) = \text{sign(det(}\mathbf{V}))\cdot\mathbf{P}_{\alpha} \sigma^{\alpha}$$

with $\mathbf{P}$ a positive reorientation of momentum so that it is the Hamiltonian for a Weyl fermion of handedness determined by $\text{sign(det(}\mathbf{V}))$

Now $H$ is a map from the p-space $\mathbb{T}^3$ to the 3-dimensional real-space of traceless hermitian 2x2 matrices $\mathfrak{su}(2)$.

In what follows, we will use the concept of manifold: a manifold of dimension n is defined as a a space everywhere locally homeomorphic to $\mathbb{R}^n$.

We place ourselves in the 6-dimensional manifold $\mathbb{T}^3\times\mathfrak{su}(2)$, in which we define the 3-dimensional submanifolds  \( S := \{ (p, H(p)) \mid p \in \mathbb{T}^3 \} \) and $\mathbb{T}^3 \times \{0\}$, which are automatically oriented.

What we want to study is the intersection number\cite{1974Dt/b}:
$$\text{Int}(S, \mathbb{T}^3 \times \{0\})$$
which we define below (FIG. \ref{fig:intersection}) and show \underline{is exactly} the sum of the signs of $\text{det}(V)$ at every intersection point and so \underline{the sum of handedness of all the Weyl cones} in the Brillouin zone.  \\

The intersection number can be defined only if the dimensions of the two sub-manifolds are supplementary. Intuitively it is the algebraic number of intersections so for example in figure FIG. \ref{fig:intersection} it is 0.\\

The orientation number $\epsilon$ is defined by taking a direct basis of each sub-manifold S and S' at the intersection point and concatenate the two basis to get a basis of the embedding manifold. $\epsilon$ is set to be $+1$ if the basis is direct and $-1$ if it is indirect. Int($S,S'$) is defined to be the sum of the orientation numbers at each intersection points.

\begin{figure}[h]
    \centering
    \includegraphics[width=0.9\linewidth]{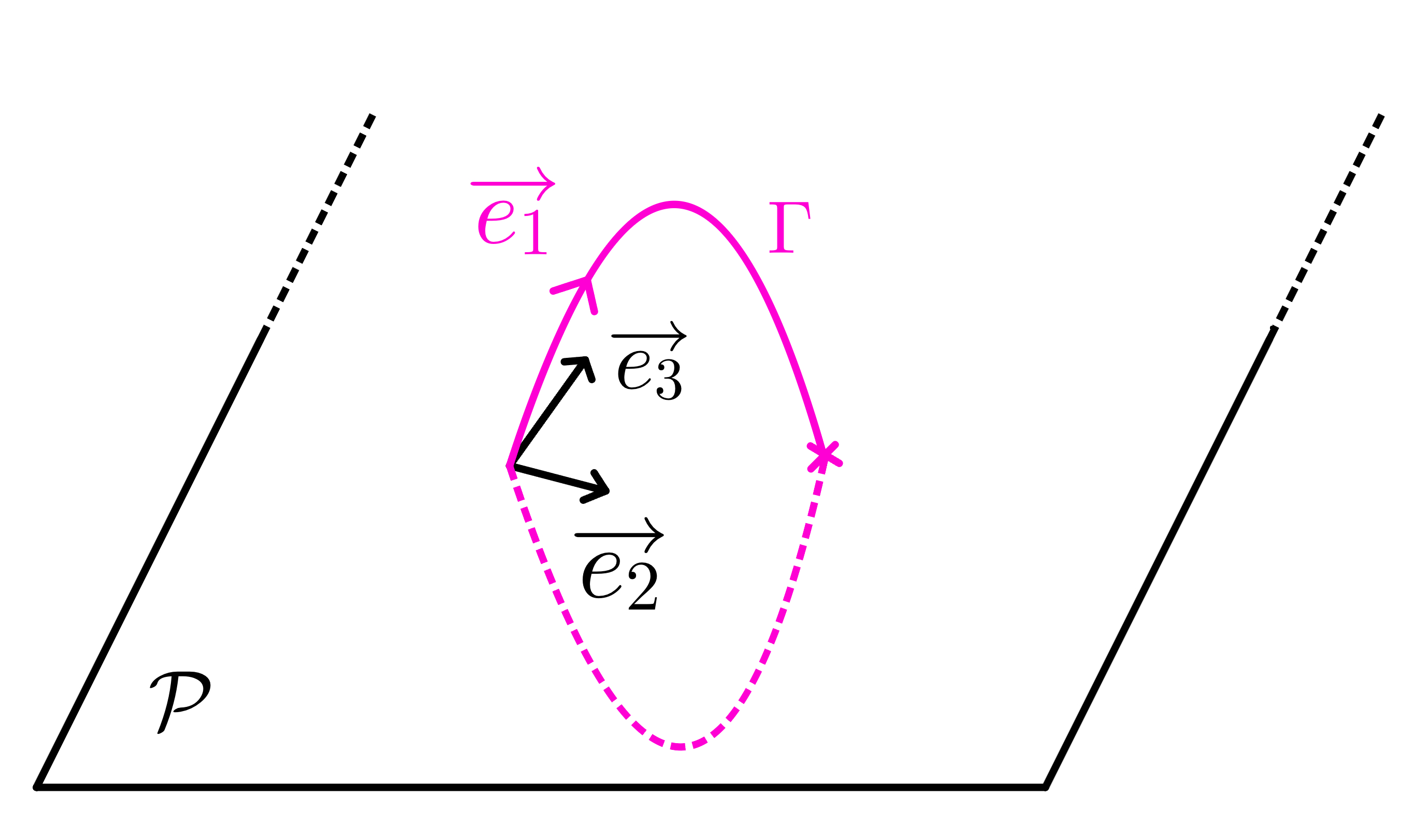}
    \caption{Example of definition of the orientation number in $\mathbb{R}^3$:  $\epsilon = \text{sign(det(}\vec{e_1},\vec{e_2},\vec{e_3}))$. Here S=$\Gamma$ is a circle and S'=$\mathcal{P}$ is a plane}

\label{fig:intersection}
\end{figure}

In our case let us calculate one orientation number:

We have $H(\mathbf{p}) = \text{sign(det(}\mathbf{V}))\cdot\mathbf{P}_{\alpha} \sigma^{\alpha}$ so we can take the following basis for S:

$$((\mathbf{P}_x, \eta \sigma^x), (\mathbf{P}_y, \eta \sigma^y),(\mathbf{P}_z, \eta \sigma^z)), \qquad \eta = \text{sign(det(}\mathbf{V}))$$

and for $\mathbb{T}^3 \times \{0\}$:

$$((\mathbf{P}_x, 0), (\mathbf{P}_y, 0),(\mathbf{P}_z,0))$$

We then concatenate the two bases and operate with transvections and transpositions:

\begin{align*}
    (&(P_x, \eta \sigma^x), (P_y, \eta \sigma^y),(P_z, \eta \sigma^z),\\
    &(P_x, 0), (P_y, 0),(P_z,0))\\
    \rightarrow (&(0, \eta \sigma^x), (0, \eta \sigma^y),(0, \eta \sigma^z),\\
    &(P_x, 0), (P_y, 0),(P_z,0))\\
    \rightarrow -&((P_x, 0), (P_y, 0),(P_z,0),\\
    &(0, \eta \sigma^x), (0, \eta \sigma^y),(0, \eta \sigma^z)
\end{align*}

So the orientation number is $-\eta = -\text{sign(det(}\mathbf{V}))$ so the opposite of the handedness of the Weyl cone. We will show that the sum of these numbers is 0:

It is a known property that the intersection number is invariant by homotopy \cite{1974Dt/b}, a proof is given in \ref{Part_A_Intersection}. Then define the homotopy:

$$F(s) = (1-s)S + s (\mathbb{T}^3\times\{A\}), \qquad A = \begin{pmatrix}
    1&0\\
    0&-1
\end{pmatrix}$$

So that $F(1) = \mathbb{T}^3\times\{A\}$, then: $$\text{Int}(S,\mathbb{T}^3\times\{0\}) = \text{Int}(\mathbb{T}^3\times\{A\},\mathbb{T}^3\times\{0\}) = 0$$

Because $A$ has non degenerate eigenvalues.
This finishes the proof in the case of a 2-band model because the only quantized charge we can define is identity (if the eigenvalues are fixed to one at some Weyl point, then they are fixed to one everywhere).
\end{proof}

\subsection{The case of real fermions}

In the real fermion formalism, the implications of the Nielsen-Ninomiya no-go theorem manifest differently due to the intrinsic particle-hole (charge conjugation) symmetry of the system. This symmetry, which is a structural feature of the real fermion representation, constrains the possible configurations and chirality of band crossings.

To understand this, consider the constraint derived earlier (Eq.\ref{eq:H_real}):

$$H_\text{Real}(\mathbf{k}) = -H^*_\text{Real}(-\mathbf{k})$$

This condition enforces that for any band crossing at momentum $\mathbf{k}_0$, there exists a corresponding crossing at $-\mathbf{k}_0$, with opposite chirality. Consequently, the real fermion representation automatically pairs each Weyl cone with its mirror image under inversion of momentum, and these paired crossings have opposite topological charges.

This enforced pairing structure due to charge conjugation predicts, independently from theorem \ref{thm: NN}, that the minimal number of real fermion band crossings is two, appearing at $\pm \mathbf{k}_0$. Each crossing in such a pair carries opposite chirality (e.g., left-handed at $\mathbf{k}_0$ and right-handed at $-\mathbf{k}_0$).

However, this also reveals a structural advantage: by leveraging the charge-conjugation symmetry, we can isolate a single chiral degree of freedom from a pair of real fermion cones. This becomes feasible when we reinterpret the two crossings as particle-hole partners, and either project out the redundant degrees of freedom or reinterpret the crossings as carrying opposite emergent charges. This subtle reinterpretation, unique to the real fermion framework, creates the possibility of realizing an effective single Weyl cone in the infrared, without violating the no-go theorem.

Furthermore, the parity of the number of real fermion band crossing pairs (i.e., $M = 1, 2, \ldots$ giving $2M$ crossings) determines the effective low-energy theory:

\begin{itemize}
    \item $M$ = 2, 4,\ldots : leads to multiple Dirac cones or even numbers of Weyl cones, consistent with the traditional no-go statement for complex fermions.
    \item $M$ = 1, 3,\ldots : corresponds to a more unusual low energy theory where, due to parity considerations, at least one of the pairs of real fermion crossings cannot be combined with another pair to form a Dirac fermion. In this article we are mostly interested in the case $M$ = 1 where the single pair of real fermion crossings can be manipulated to yield a single emergent Weyl fermion through projection or symmetry-based reinterpretation.
\end{itemize}

This real fermion structure thus circumvents the conventional obstruction by encoding symmetry relations that are absent in complex fermion descriptions. The enforced charge conjugation pairing in momentum space becomes a powerful tool to reinterpret lattice band structures and serves as the foundation for constructing single Weyl systems on the lattice.

\subsection{Recovering the single Weyl fermion}

In this section, we show how the foundation we just derived for real fermions allows for the explicit reconstruction of a single Weyl fermion by projection using the Schrieffer-Wolff transformation to break the redundancy imposed by real fermion symmetry.

\subsubsection{Weyl fermion in the real fermion formalism}

Let us first point out the real fermion expression of the Weyl fermion.
In the standard complex formalism, a Weyl fermion with a given Handedness can be assigned with the following Hamiltonian

\begin{equation}
    \label{eq: CWeyl}
    H_{\text{complex}}(\mathbf{k}) = \mathbf{k}\cdot\sigma
\end{equation}

After transformation to real fermion formalism, a Weyl fermion is expressed as:

\begin{equation}
    \label{eq: RealWeyl}
    H_{\text{Real}}(\mathbf{k}) = 
        k_x\sigma^x+k_z\sigma^z -k_y\tau^y\otimes\sigma^y
\end{equation}

\subsubsection{Projecting out the excessive degrees of freedom using the Schrieffer-Wolff transformation}
\label{part: SW}


Now we begin with Hamiltonian (Eq.\ref{eq: NodalH}) which can describe a fully gapped SPT or in this case a DIII class topological superconductor subject to a magnetic field. The band structure in a strong magnetic field is shown in FIG. \ref{fig:Band_structure_real_weyl}. The effect of a strong magnetic field is to first isolate low energy degrees of freedom so that we can reconstruct a single Weyl fermion. To do so, we first discard the two higher-energy bands (represented by blue dotted lines), retaining only the lowest-energy bands. We then re-express the remaining degrees of freedom in a form characteristic of a Weyl fermion.

To achieve this, we employ the Schrieffer-Wolff (SW) transformation, a unitary transformation designed to approximately block-diagonalize a Hamiltonian by integrating out high-energy states perturbatively. This transformation is particularly effective for separating low-energy effective theories from higher-energy excitations. We leave the calculations in appendix \ref{appendix:SW}.

Now we can zoom in to the region near $\pm {\bf k}_{deg}$ (the subscript {\em deg} refers to the degeneracy points) where the two middle bands become degenerate, i.e. cross each other (See Fig.\ref{fig:Band_structure_real_weyl}).

After transformation, we end up with a 2-band Hamiltonian $H_{red}$, which writes:

\begin{equation}
\label{eq:Hred}
\begin{aligned}
    H_{red}(\mathbf{k+k_{deg}}) &= \mathbf{k}_\kappa V^\kappa\cdot\sigma\\
    H_{red}(\mathbf{k-k_{deg}}) &= \mathbf{k}_\kappa V^\kappa\cdot\sigma^T
\end{aligned}
\end{equation}

To recover the Weyl fermion, we only need a few last tricks. We go to the IR limit, let $\chi_R (\mathbf{k})= \chi (\mathbf{k}+\mathbf{k}_{deg})$, $\chi_L (\mathbf{k})= \chi (\mathbf{k}-\mathbf{k}_{deg})$ which are now 2-components real fermions, and define a new 4-component fermion by letting:
\begin{equation}
    \chi' = \begin{pmatrix}
        \chi_R \\ 
        \chi_L
    \end{pmatrix}
\end{equation}
The Hamiltonian for this 4-component fermion is then:
\begin{widetext}
\begin{equation}
    \mathcal{H}_{\text{eff}} = \int d^{3}k \, \chi'(-\mathbf{k}) \left( P_x\sigma^x+P_z\sigma^z +P_y\tau^z\otimes\sigma^y \right)\chi'(\mathbf{k})
\end{equation}
where we redefined momentum space by setting $\mathbf{P} = \mathbf{k}_\kappa V^\kappa$.
\end{widetext}

Notice that $\chi_R^{\dagger}(x) =  \chi_L^T(x)$ (because $\chi_R^{\dagger}(\mathbf{k}) =  \chi^T_L(-\mathbf{k})$), so by this concatenation process, we lost the reality condition. We can recover it by applying the transformation (Eq.\ref{eq:R_trans}) once again \cite{Kapoor}:

\begin{equation}
\label{eq: trans_majorana}
    \chi(x) :=  \frac{1}{\sqrt{2}} \begin{pmatrix} I_2 & I_2 \\ -iI_2 & iI_2 \end{pmatrix}  \chi'(x)
\end{equation}
\begin{equation}
    R = \frac{I}{2} - \frac{i}{2}(\tau^x -\tau^y +\tau^z)
\end{equation}

The form of $H_{red}$ around the degeneracy points (Eq.\ref{eq:Hred}) then automatically give us the formulation of a single Weyl fermion in the real fermion formalism. That is:
\begin{equation}
    H_{\text{red/Real}}(\mathbf{k}) = 
        P_x\sigma^x+P_z\sigma^z -P_y\tau^y\otimes\sigma^y
\end{equation}
therefore describing a right or left-handed Weyl fermion depending on the sign of $\text{det}(V)$. The chiral charge operator in the real fermion formalism therefore takes the value:
\begin{equation}
    Q_\text{chiral} = -i\sigma^x\sigma^z(-\tau^y\otimes\sigma^y) = \tau^y.
\end{equation}
The two eigen values of this operator, $\pm 1$ can be related to particle-like and anti-particle-like states of a Weyl fermion.

Finally, let us remark here that for this part of discussion, we have swapped $\tau$ (usually for the charge space) and $\sigma$ (usually for the spin space) to illustrate the relation to the Weyl fermion.
More precisely, the model above forms a $\tau$-$\sigma$ dual 
of the Weyl fermion model as discussed previously\cite{Zhou1}.

Before leaving this section, we want to mention an explicit construction of a Dirac fermion out of a $(2+1)D$ square lattice Hubbard model of real fermions by Affleck et al. in Ref.\cite{Affleck17}. Although our current application of a theorem of the intersections of two sub-manifolds and later discussions in the article are exclusively on $(3+1)D$ lattice models, a pair of crossing points in real fermions related to the charge conjugation symmetry also play a crucial role in that concrete construction in $(2+1)D$.

\section{Fermion Lattice Model Analysis I}
\label{part: Model_analysis}

We now apply the real fermion framework to concrete lattice models of gapless superconductors that exhibit single Weyl cone dynamics. By following the three paths introduced earlier, in this and next section, we present a few representative models and analyze their symmetry properties via constructing exactly conserved charges. We will also discuss the relations to the exact symmetries in recent proposals of chiral lattice fermions\cite{Chatterjee_2025,gioia2025exactchiralsymmetries31d}.

In the Sect.\ref{part: family}, We also demonstrate that all these lattice models can be further unified under a common structure governed by two dual copies of $3\otimes3$ representations of an $Spin(4)$ group.

 The lattice models to be studied illustrate our three main paths towards single Weyl fermion are summarized below:
 
Path \textbf{a}): \\
 A simple lattice model of a time reversal symmetric tQCP with protection symmetry $G_p=Z_2^T$, the time reversal symmetry group. And we restrict ourselves to 
 a generic tQCP where the change of topologies is minimum. The change of band toplogies here physically is related to the recontruction of surface majorana cones and dictates the gapless fermions at tQCPs (see \ref{FDF1})\footnote{ Here We label the majorana band topologies here via $N_w=2n, n=0,1,2,...$. In literature, the band topologies sometimes are also labeled by $\nu$ which is related to $N_w$ via $N_w=2\nu$}.
 
 The change takes a fundamental value consistent with the protecting symmetry $G_p$ which in this case is
 $\delta N^f_w=2$ (the superscript refers to the fundamental value of changes).
 The specific lattice model describes a tQCP in the DIII class topological superconductors.
The IR limit of such a tQCP has been studied quite extensively in Ref.\cite{Zhou1,Zhou2} and its lattice model in a connection to a 3d boundary of a 4d topological state was also presented\cite{Zhou3}.

Path \textbf{b}):\\
 Applying a strong magnetic field to a fully gapped SPT again in the DIII class with protecting symmetry $G_p=Z_2^T$, i.e. a topological superconductor or superfluid that break the charge $U(1)$ symmetry. Quantum phase transitions into these gapless superconducting nodal points 
 in Ref.\cite{Yang_Zhou,Zhou1} and a lattice model was presented in a unpublished thesis\cite{Kapoor}.

Path \textbf{c}): \\
 The last three models involve either a generic tQCP or fine-tuned multi-critical tQCPs which are further subject to applications of time reversal symmetry breaking actions. It can be thought as a hybrid of Path \textbf{a}) and Path \textbf{b}). In Type \textbf{c}), we study lattice models of tQCPs with the change of topologies taking values of integer multiples of the fundamental values (i.e. $\delta N^f_w=2$), i.e. we work with tQCP models with $\delta N_W=4,8$ in additional to a generic tQCP with $\delta N_w=\delta N^f_w=2$. 
A lattice model that was originally proposed a few months ago in \cite{gioia2025exactchiralsymmetries31d} belongs to one of Type \textbf{c} models in our classification with $\delta N_w=8$. 

We will focus on the UV completion of the emergent infrared single Weyl cone physics in gapless superfluids and corresponding conserved charge operators in lattice models in the 3-torus momentum space, $\mathbb{T}^3$.

Before starting detailed discussions, let us point out that the number of degrees of freedom
(normalized in terms of 3D Dirac fermions), $N_D$ (the subscript D here refers to Dirac fermions), appears at a tQCP in fermionic SPTs in general depends on both $G_p$, the protecting symmetry $G_p$ and $\delta N_w$, the change of topologies across a tQCP.

We can introduce such a general relation as

\begin{eqnarray}
N_D=N_D(G_p, \delta N_w)= N^f_D(G_p) \frac{\delta N_w}{\delta N^f_w}.
\label{FDF1}
\end{eqnarray}
Here $N^f_D(G_p)$ is the number of fundamental degrees of freedom when the change of topologies takes a fundamental value. $\delta N^f_w=\delta N^f_w(G_p)$ again is the fundamental value of the change of topologies
defined by the protecting symmetry $G_p$. $\delta N_w$ is the actual change of topologies at a tQCP which interests us.

For the cases we are interested with $G_p=Z^T_2$ being the time reversal symmetry, $N^f_D(G_p=Z^T_2)=\frac{1}{2}$ which is equivalent to the degree of freedom of a Weyl fermion, as being studied extensively in Ref. \cite{Zhou1,Zhou1,Zhou3}. This aspect has also been utilized recently to understand dynamic critical exponents $z$ in weakly and strongly interacting tQCPs in topological states\cite{Yang2025}. 

Therefore in our current discussions, Eq.\ref{FDF1} indicates that

\begin{eqnarray}
N_D=\frac{1}{4} {\delta N_w}
\label{FDF2}
\end{eqnarray}
which implies that $N_D=1,2$ when $\delta N_w=4,8$ as $\delta N^f_w(Z^T_2)=2$.

Only at a generic tQCP with a minimum change of topologies $\delta N=2$, 
$N_D=\frac{1}{2}$ exactly matches the degree of freedom of a Weyl fermion which was observed in previous studies of tQCPs.
And this is also the main reason that, for our purspose, at tQCPs with $\delta N_w$ being integer multiples of the fundamental values $\delta N^f_w=2$ such as $\delta N_w=4,8$, time reversal symmetry breaking actions are always needed to further reduce the number of degrees of freedom $N_D$ to $N_D=\frac{1}{2}$, the degree of freedom in a single Weyl fermion. 

As we have predicted before, and will be further seen below, in both Path \textbf{b}) and Path \textbf{c}), the time reversal symmetry has to be broken in the lattice models, in addition to the charge $U(1)$ spontaneously breaking the symmetry to make the construction feasible. Only in Path \textbf{a}), the T-symmetry can be fully preserved. 

In this section, we will first focus on Path \textbf{a}) and Path \textbf{b}), and illustrate the ideas via two simplest lattice models, respectively. The infrared properties of gapless superfluids in these models have been previously studied and are known. And, as stated before, here we will focus on their UV completion and 
the UV completed symmetry group.

\subsection{ Lattice Model I: \\A tQCP approach with $\delta N_w=\delta N^f_w(G_p)=2$}



This model follows path \textbf{a}) discussed above. It was previously introduced for a tQCP in a DIII class topological superconductors where strong interactions are present and there is an emergent infrared space-time Lorentz symmetry due to strong coupling\cite{Yang_Zhou,Zhou1,Zhou1,Zhou3}. Practically it can be applied to describe a tQCP in a 3d time-reversal symmetric p-wave superconductor or, isomorphically, a $^3$He superfluid phase\cite{Leggett75,Volovik03}.

We introduce an elementary lattice model of a T-invariant tQCP with protection symmetry $G_p=Z_2^T$ and a minimal change of topology of $\delta N_W = 2$, i.e. the change of topologies in this model takes a fundamental value $\delta N^f_w(G_p)=2$. As discussed in section \ref{part: 2_general_single_weyl}, the T-invariant tQCP can be obtained by setting the mass parameter below $\epsilon$ to 0:
\begin{equation}
\label{eq: tQCP}
\begin{aligned}
H(\mathbf{k}) = 
& \sin(k_x) \tau^z \otimes \sigma^x +  \sin(k_y) \tau^z \otimes \sigma^z +  \sin(k_z) \tau^x\\ &+ \left(\mu-\sum_i \cos(k_i) \right) \tau^y
\end{aligned}
\end{equation}
where $\mu=3+\epsilon$.

As will be discussed in section \ref{part: Charge}, due to the bands being everywhere Kramer-degenerate, there are many conserved symmetry charges one can define for this model. Here is the most symmetric one:
\begin{equation}
    \begin{aligned}
    Q_1(\mathbf{k}) = &\sin(k_x) \sigma^z - \sin k_y \sigma^x - \sin k_z \tau^y \otimes\sigma^y \\
    &- \frac{\sum_i \sin^2(k_i)}{\sum_i \sin^2(k_i/2)}\tau^x \otimes \sigma^y.
\end{aligned}
\label{eq: Q1}
\end{equation}

The symmetry charge flows into a desired structure of $$-\tau^x \otimes \sigma^y$$ in the infrared limit of ${\bf k}\rightarrow 0$. This is an emergent chiral charge for a tQCP\cite{Zhou1}. As predicted for any chiral charge of a single Weyl fermion, the charge goes to 0 at $\mathbf{k} = (\pi,\pi,\pi)$. It is possible to find the real space charge by expanding the last term in a Fourier series:

$$c_{\bf r} =  \int_{BZ} \frac{\sum_i \sin^2(k_i)}{\sum_i \sin^2(k_i/2)}e^{-i\bf r\cdot \bf k} d\bf k$$

Then the real space charge writes:

\begin{equation}
    \begin{aligned}
        \mathcal{Q}_1 = \sum_{\bf i} [&\chi_{{\bf i}+\hat{x}}i\sigma^z\chi_{\bf i} -\chi_{{\bf i}+\hat{y}}i\sigma^x\chi_{\bf i} -\chi_{{\bf i}+\hat{z}}i\tau^y\otimes\sigma^y\chi_{\bf i}\\
        & -\sum_{\bf r} c_{\bf r} \chi_{{\bf i}+{\bf r}}\tau^x\otimes\sigma^y\chi_{\bf i}] +h.c.
    \end{aligned}
\end{equation}
with ${\bf r}$ going through the whole lattice. We use $\hat{x},\hat{y},\hat{z}$ to denote the neighboring sites along the $x,y,z$-directions respectively.

This charge has a complicated real-space spatial structure because of the k-dependence in Eq.\ref{eq: Q1}. 
Alternatively, we can have the following charge:

\begin{equation}
    \begin{aligned}
    Q_2(\mathbf{k}) = &\sum_i \sin^2(k_i/2) \times \left(\sin(k_x) \sigma^z - \sin k_y \sigma^x - \sin k_z \tau^y \otimes\sigma^y \right)\\
    &- \sum_i \sin^2(k_i)\tau^x \otimes \sigma^y
\end{aligned}
\end{equation}
which has a simpler real space expression with second nearest neighbor couplings:
\begin{equation}
    \begin{aligned}
    \mathcal{Q}_2 = \sum_{\bf i}&[\left(3\chi_{{\bf i}+\hat{x}}-\frac{1}{2}(\sum_{\hat{b}} \chi_{{\bf i}+\hat{x}+\hat{b}}+\chi_{{\bf i}+\hat{x}-\hat{b}} ) \right) i\sigma^z \chi_{\bf i} \\&-\left(3\chi_{{\bf i}+\hat{y}}-\frac{1}{2}(\sum_{\hat{b}} \chi_{{\bf i}+\hat{y}+\hat{b}}+\chi_{{\bf i}+\hat{y}-\hat{b}} ) \right) i\sigma^x \chi_{\bf i} \\&-\left(3\chi_{{\bf i}+\hat{z}}-\frac{1}{2}(\sum_{\hat{b}} \chi_{{\bf i}+\hat{z}+\hat{b}}+\chi_{{\bf i}+\hat{z}-\hat{b}} ) \right) i\tau^y\sigma^y \chi_{\bf i} \\&- \left(3\chi_{\bf i}- \chi_{{\bf i}+2\hat{x}}- \chi_{{\bf i}+2\hat{y}}- \chi_{{\bf i}+2\hat{z}}\right) \tau^x\sigma^y\chi_{\bf i}]+h.c.
\end{aligned}
\end{equation}
with $\hat{b}$ only going through the neighboring sites.

\subsection{Lattice model II:\\ A gapped SPT with strong Time-reversal-symmetry breaking fields}

This model follows path {\bf b)}. It can applied to study a fully gapped SPT in DIII class but further subject to a strong magnetic field\cite{Yang_Zhou,Kapoor}.

We add a magnetic field to a gapped SPT with Time-Reversal symmetry (corresponding to the previous model with $\epsilon \neq 0$), therefore producing a superconductor or superfluid that breaks the charge U(1) symmetry. We choose a particular representation of the $\Gamma$ matrices of section \ref{part: 2_general_single_weyl}, so that the Hamiltonian looks like a lattice completion of a p-wave superconductor Hamiltonian (which presents explicit charge U(1) symmetry breaking). The model studied is the following:
\begin{widetext}
\begin{equation}
\label{eq:H_p-wave}
H (\mathbf{k})= - \sin(k_x) \tau^z \otimes \sigma^z +  \sin(k_y) \tau^x \otimes I_2+ \sin(k_z) \tau^z \otimes \sigma^x + \left( \mu - \sum_i \cos(k_i) \right) \tau^y \otimes I_2 +B \tau^y \otimes \sigma^z
\end{equation}
where $\mu$ is chosen so that it is a fully gapped SPT when $B=0$ (See Fig.\ref{fig:phase diagram}). This model was previously proposed and studied in Ref.\cite{Kapoor}.

This $k$-space Hamiltonian can be obtained from a real space Hamiltonian:
\begin{equation}
    \begin{aligned}
        \mathcal{H} = 
 \mu \sum_{\bf i} c_{\bf i}^{\dagger} c_{\bf i} 
- \frac{1}{2} \sum_{\bf i} \left\{ 
     c_{\bf i}^{\dagger} c_{{\bf i}+\hat{x}} 
    +  c_{\bf i}^{\dagger} c_{{\bf i}+\hat{y}} 
    +  c_{\bf i}^{\dagger} c_{{\bf i}+\hat{z}} 
    + \text{h.c.} 
\right\} 
+ \sum_{\bf i} c_{\bf i}^{\dagger} \left( \vec{B} \cdot \vec{\sigma} \right) c_{\bf i} \\
- \frac{1}{2} \sum_{\bf i} \left\{ 
     c_{\bf i} (\sigma^z) c_{{\bf i}+\hat{x}} 
    + i  c_{\bf i} c_{{\bf i}+\hat{y}} 
    -  c_{\bf i} (\sigma^x) c_{{\bf i}+\hat{z}} 
    + \text{h.c.} 
\right\}
    \end{aligned}
\end{equation}
\end{widetext}

Let us find the charge for this system. We have a commuting operator:
\begin{equation}
    \label{eq: Q_pwave}
    Q(\mathbf{k}) = f(\mathbf{k})\left[ \sin(k_z) \tau^x \otimes \sigma^y-( \mu - \sum_i \cos(k_i)) I \otimes \sigma^z\right]
\end{equation}
with $f$ an odd function of momentum to be consistent with charge conjugation. This implies that the symmetry group is non-compact as the charge is a continuous function and has to be 0 at $\mathbf{k}=0$.

By setting ${\bf k}=\pm {\bf k}_{deg}$ at the band crossing points $\pm {\bf k}_{deg}$, one finds the non-zero symmetry charges at the two crossings respectively.

\begin{figure}[h]
    \centering
    \includegraphics[width=\linewidth]{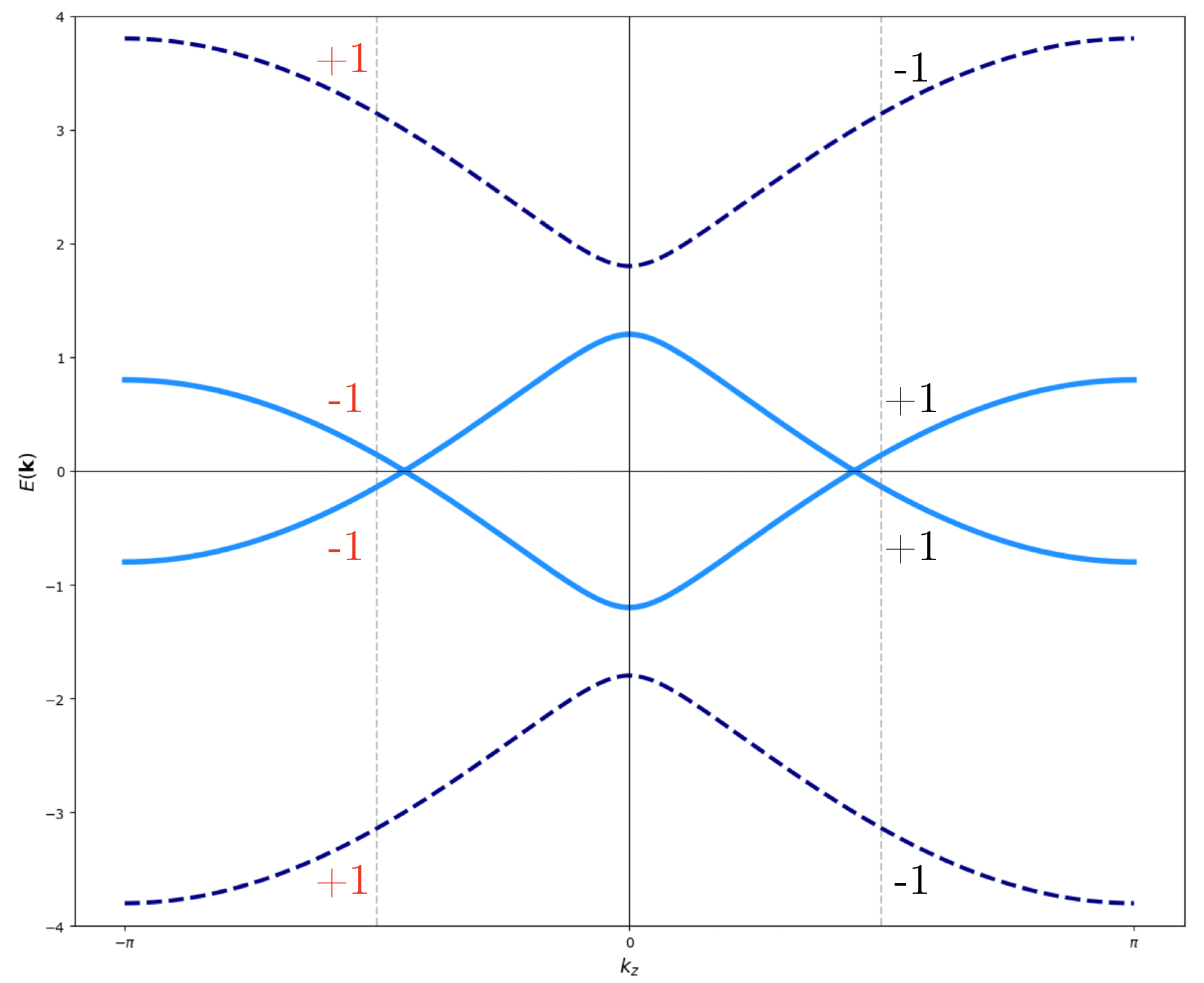}
    \caption{Charges (in {\it arbitrary units}) in the solid blue (dotted purple) band vary from +1 (-1) to -1 (+1) when $\mathbf{k}$ varies from $+\mathbf{k}_{deg}$ to $-\mathbf{k}_{deg}$ following Eq.\ref{eq: Q_pwave} with $f(\mathbf{k}) = \sin(k_z)$.}
    \label{fig:charge_nodal2}
\end{figure}

In Fig.\ref{fig:charge_nodal2}, we illustrate how the charges assigned to each of four bands vary according to the charge operator defined in Eq.\ref{eq: Q_pwave}. Later in the next section and appendix F, we will show that such charge assignments are always implementable and they are the direct consequence of charge conjugation symmetry.

We further choose $f(\mathbf{k}) = \sin(k_z)$ to have a simple real space charge that involves only second-neighbor coupling:
\begin{equation}
\begin{aligned}
    \mathcal{Q} = \sum_{\bf i}&[\left(\frac{1}{2}(\sum_{\hat b} \chi_{{\bf i}+\hat{z}+{\hat b}}+\chi_{{\bf i}+\hat{z}-{\hat b}} ) -\mu\chi_{{\bf i}+\hat{z}}\right) iI\otimes\sigma^z \chi_{\bf i} \\&+ \frac{1}{2}\left(\chi_{\bf i}-\chi_{{\bf i}+2\hat{z}}\right) \tau^x\otimes\sigma^y\chi_{\bf i}]+h.c.
\end{aligned}
\end{equation}

As predicted by the second no-go theorem, the action of the symmetry group on an on-site operator involves neighboring sites and is clearly non-local.

\section{Fermion Lattice Model Analysis II}
\label{part: Model_analysisII}

In the previous section, we have outlined the two independent approaches to single Weyl-cone dynamics. One (Path a) relies on topological quantum critical points (tQCPs) but with time reversal symmetry while the second one (Path b) is via applying a reversal symmetry breaking field to a generic gapped symmetry protected state of topological superconductors. 

It is also possible to have a hybrid approach by further applying time reversal symmetry breaking fields to tQCPs. As discussed before at the beginning of the previous section, tQCPs can be characterized by the change of global topologies $\delta N_w$ and the corresponding degrees of fermion freedom $N_D=\frac{1}{4}\delta N_w$. Below we will discuss 
three lattice models corresponding to $\delta N_w=2,4,8$ respectively.
Among them, $\delta N_w=8$ is very special as it is directly related to a recently proposed lattice model for chiral fermions.

The model of $\delta N_w=2$ with $N_D=\frac{1}{2}$ forms the fundamental
representation of tQCPs in SPTs with the protecting symmetry group $G_p=Z^T_2$, i.e. time reversal symmetric with $T^2=-1$.
It has been the focus of a few previous studies by one of the authors.
We will start with this simplest limit and discuss what happens when an additional magnetic field is applied.

\subsection{Model III: a tQCP with $\delta N_w=\delta N^f_w(G_p)=2$ further subject to T-symmetry breaking actions}

This model is closely related to Eq.\ref{eq: tQCP} with $\epsilon=0$ but further with a time reversal symmetry breaking magnetic field $B$,

\begin{equation}
\label{eq: tQCP2B}
\begin{aligned}
H(\mathbf{k}) = 
& \sin(k_x) \tau^z \otimes \sigma^x +  \sin(k_y) \tau^z \otimes \sigma^z +  \sin(k_z) \tau^x\\ &+ \left(3-\sum_i \cos(k_i) \right) \tau^y +B I \otimes \sigma^y.
\end{aligned}
\end{equation}
The result of a B-field here is to further break the two-fold degeneracy of the Hamiltonian. 

The number of Weyl fermions in Eq.\ref{eq: tQCP2B} depends crucially on 
the amplitude of $B$. Of particular interest to us is 
when 
\begin{eqnarray}
 i) 0< B < 2; &\mbox{  } ii) 4< B< 6.   
\end{eqnarray}
Eq.\ref{eq: tQCP2B} leads to one single pair of isolated crossing points along the $k_z$ axis at $\pm k_0$ where $$2\sin \frac{k_0}{2} =B$$.

A conserved charge for this model is:
\begin{equation}
\label{eq: Q_modelIII}
    Q(\mathbf{k}) =\sin^2(k_z)\tau^x\otimes \sigma^y  +\sin(k_z) [ 3-\sum_{i=1,2,3} \cos(k_i) ] \tau^y\otimes \sigma^y.
\end{equation}
\begin{figure}[h]
    \centering
    \includegraphics[width=\linewidth]{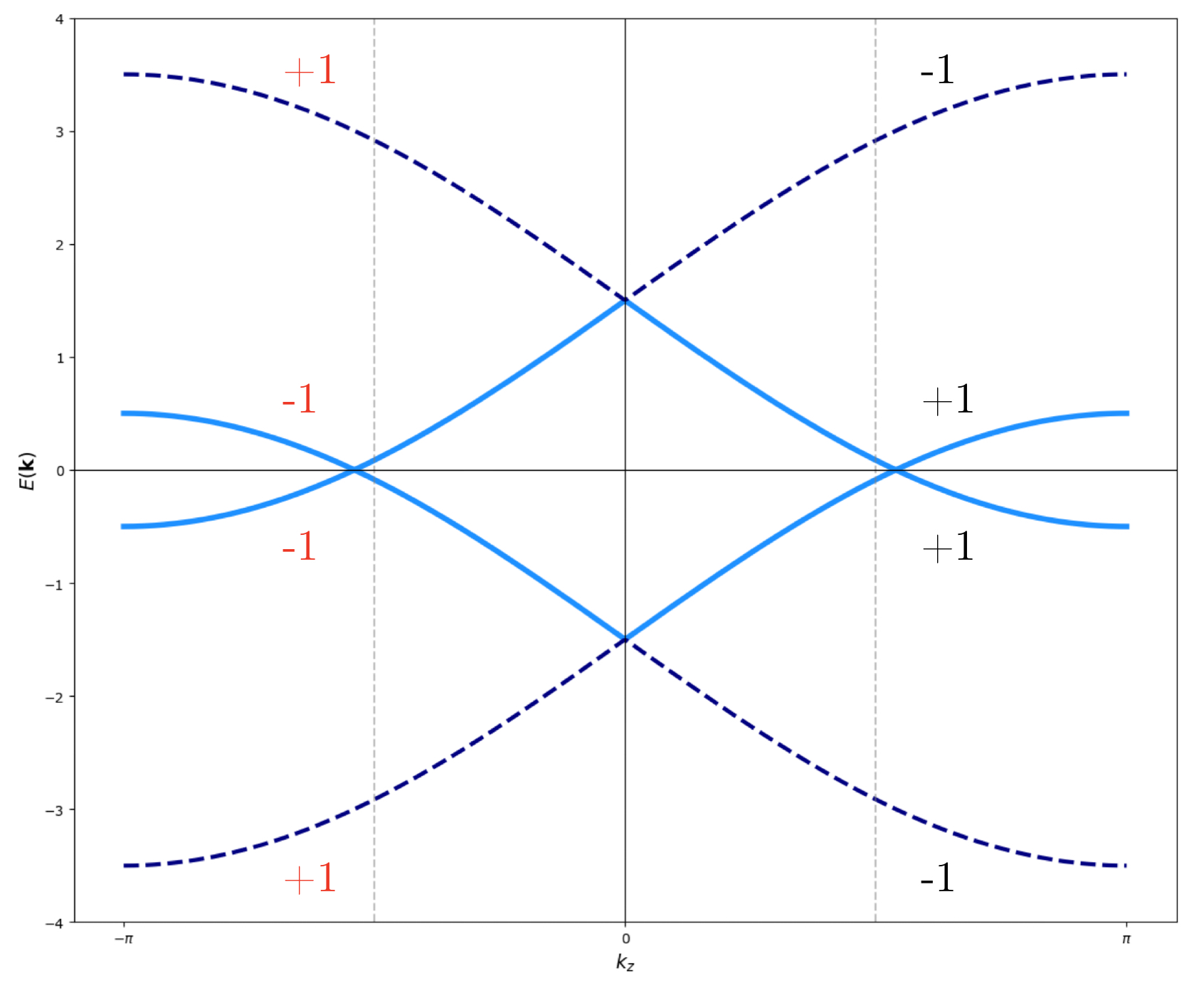}
    \caption{Charges (in {\it arbitrary units}) in the solid blue (dotted purple) band vary from +1 (-1) to -1 (+1) when $\mathbf{k}$ varies from $+\mathbf{k}_{deg}$ to $-\mathbf{k}_{deg}$ following Eq.\ref{eq: Q_modelIII} with $f(\mathbf{k}) = \sin(k_z)$.}
    \label{fig:charge_nodal3}
\end{figure}
The infrared limit shall again be taken with care.
The symmetry charge at crossing points ${\bf k}=\pm k_0{\bf e}_z$ are non-zero.
As ${\bf k} \rightarrow \pm k_0{\bf e}_z$, 
$$Q({\mathbf{k}})=Q(\mathbf{k}=\pm k_0{\bf e}_z).$$
This defines the emergent single Weyl fermion.

On the other hand, as ${\bf k}\rightarrow 0$,
$$Q(\mathbf{k})= k_z^2 \tau^x \otimes \sigma^y+...,$$ where we have muted less relevant terms. Note that in this case, a gap opens up at ${\bf k}=0$; the infrared physics here is instead determined by two crossings at ${\bf k}=\pm k_0{\bf e}_z$.

The charge-conjugation symmetry puts a severe general constraint on the charges that can be assigned to a pair of two states at any $\pm {\bf k}$ points if they can be transformed into each other by a charge conjugation transformation. To illustrate this, we first present two general statements about conserved charges following the charge conjugation symmetry.

Theorem \textbf{A}: At any momentum ${\bf k}$, it is always possible to assign the same charge of $+1$ (up to a multiplcation factor) to any \textit{two} out of four real fermion bands of our interests, disregarding their energy eigenvalues $E_\alpha(k), \alpha =1,2,3,4$ of the real fermion bands.

This property follows directly the two- or even higher-dimensional linear space spanned by symmetry charge operators in the lattice models. It is an outcome of multiple UV completed symmetry charges.

Theorem \textbf{B}: The charges assigned to a pair of states at $\pm {\bf k}$ shall always be of the same magnitude but precisely opposite to each other, if these two states transform into each other under the charge conjugation symmetry. This reflects a very generic aspect of charge conjugation symmetry.
Detailed proof of these two theorems is presented in the appendix F.

The direct consequences of the above general theorems on our discussions of Weyl cones are two-folded.
First, at any two band crossing points such as the one at $k_z=k_0$ discussed above, one can always choose to assign the same charge, say $+1$, to any two bands that are crossing. This follows Theorem \textbf{A}.

Second, the charge conjugation symmetry indicates that there shall be another two-band crossing point at $k_z=-k_0$. The two crossing points at $k_z=\pm k_0$ are precisely related by charge conjugation transformation so that the charge assigned to two crossing bands at $k_z=-k_z$ shall be precisely $-1$ following Theorem \textbf{B}.

To summarize, the charges assigned to two crossing points $\pm k_0$ along the $k_z$ are of the same magnitude but precisely opposite to each other reflecting the generic feature of charge conjugation symmetry and UV completed charge symmetry group.

This feature again indicates that it is always possible to make an emergent Weyl cone of complex fermions out of a pair of real fermion crossing points, as one appears to be particle-like and the other can be exactly attributed to its hole-like counter-part. 

This is also the elementary feature in the charge assignment in the gapless nodal phase in Sect.\ref{part: Model_analysis}. The differences are in the nodal phase models, a pair of two crossing points that are transformed into each other belong the {\em same} two bands so that the charges assigned to two specific bands change their signs when moving across $k=0$ point. In the situation discussed here, the crossing at $k_0$ occur in
a pair of bands that are different from the pair of bands where the crossing occurs at its charge conjugation point$-k_0$.
So that within the same pair of two bands the charge can be of the same sign although at $k=0$ all charges vanish (See Fig.\ref{fig:charge_nodal3}).

\subsection{Model IV: a tQCP with $\delta N_w=4$ subject to T-symmetry breaking actions }

Here we introduce the second tQCP model along the path \textbf{c}). It involves a combination of path {\bf a)} and path {\bf b)}. We follow path {\bf a)} by applying a T-symmetry preserving coupling, which gaps all but the {\it two} double-degenerate real fermion crossings in the $k_x=k_y =0$ axis. It therefore corresponds to a multi-critical T-invariant tQCP with a change of topology $\delta N_W = 4$:

\begin{equation}
\begin{aligned}
    H (\mathbf{k}) = & \sin(k_x) \tau^z\otimes\sigma^x + \sin(k_z) \tau^z\otimes\sigma^z - \sin(k_y) \tau^x \\
& +(2-\cos(k_x)-\cos(k_y)) \tau^y.
\end{aligned}
\label{eq: tQCP4a}
\end{equation}

Then we follow path {\bf b)} by adding a T-symmetry breaking field, which gaps the remaining Weyl cone at $k_z =\pi$:
\begin{equation}
\label{eq:mag_gapping_1}
\begin{aligned}
    H_B (\mathbf{k}) = H (\mathbf{k})+B(1-\cos(k_z)) \sigma^y
\end{aligned}
\end{equation}



After an analysis, we find that for the range of parameter $B\in [0,2]$, we have 4 crossings (in the real fermion formalism), with critical points at $B=1$ and $B=2$.
And if and only if 
\begin{equation}
B>2, 
\end{equation}
we are left with only one Weyl cone at $k_x=k_y=k_z=0$.
So a strong magnetic field in this case indeed induces an emergent single Weyl fermion.

At last, following a similar discussions in Ref.\cite{Zhou1}, we find that the above tQCP model can be mapped into the following lattice chiral fermion model via a standard $Spin(4)$ transformation,

\begin{equation}
\begin{aligned}
    H (\mathbf{k}) = & \sin(k_x) \sigma^x +  \sin(k_z) \sigma^z - \sin(k_y) \tau^y \otimes \sigma^y \\
& +(2-\cos(k_x)-\cos(k_y)) \tau^x \otimes \sigma^y.
\end{aligned}
\label{eq: tQCP4}
\end{equation}

The UV completed symmetry of the lattice model in Eq.\ref{eq: tQCP4} can be easily found. 
The full conserved charge for this model is:
\begin{equation}
    Q(\mathbf{k}) =\sin^2(k_y)\tau^y -\sin(k_y) (2-\cos(k_x)-\cos(k_y))  \tau^x.
\end{equation}
The infrared limit can be taken and one finds $$Q(\mathbf{k}\rightarrow 0) \rightarrow k_y^2 \tau_y+...$$ which represents a generic dispersive nature of the symmetry charges associate with emergent symmetries.


\subsection{Model V: a tQCP with $\delta N_W=8$ subject to T-symmetry breaking actions }

Another model we will present here also involves a combinations of the two paths, Path {\bf a}) and Path {\bf b}). Following path {\bf a)}, we now apply a $k_z$-dependent T-symmetry preserving coupling, which gaps {\it four} of the eight double-degenerate crossings and keeps only 4 of them in the $k_z = 0$ plane. It corresponds to a multi-critical T-invariant tQCP with a change of topology $\delta N_w = 8$:

\begin{equation}
\begin{aligned}
    H_{\text{}} (\mathbf{k}) = & \sin(k_x)\tau^z\otimes \sigma^x +  \sin(k_y) \tau^z\otimes\sigma^z - \sin(k_z) \tau^x \\
& +(1-\cos(k_z)) \tau^y.
\end{aligned}
\label{eq: tQCP8a}
\end{equation}

This model along the path \textbf{c}) can be closely related to what has been studied recently in the literature.
In order to make an explicit contact with the chiral fermion model in Ref.\cite{gioia2025exactchiralsymmetries31d}, here we choose to work with an equivalent lattice model for this gapless tQCP after performing an $Spin(4)$ unitary transformation. The resultant model is

\begin{equation}
\begin{aligned}
    H_{\text{}} (\mathbf{k}) = & \sin(k_x) \sigma^x +  \sin(k_y) \sigma^z - \sin(k_z) \tau^y \otimes \sigma^y \\
& +(1-\cos(k_z)) \tau^x \otimes \sigma^y.
\end{aligned}
\label{eq: tQCP8}
\end{equation}

Then, following path {\bf b}), one can add a T-breaking field, which gaps all but one of the remaining Weyl cones:
\begin{equation}
\begin{aligned}
    H_{B} (\mathbf{k}) = H (\mathbf{k})+ \left( 2 - \cos(k_x)-\cos(k_y) \right) \sigma^y
\end{aligned}
\end{equation}

This model previously proposed in Ref.\cite{gioia2025exactchiralsymmetries31d} presents a single Weyl fermion.
However, contrary to the T-invariant tQCP with $\delta N_w=2$, for the particular choice of TRB field made in Ref.\cite{gioia2025exactchiralsymmetries31d}, there are also band touching at $(k_x,k_y,k_z)=(\pi,0,\pi)$ and $(k_x,k_y,k_z)=(0,\pi,\pi)$.

These touching points are related to quantum critical points of Lifshitz type 
and can disturb the IR dynamics as discussed in Ref.\cite{Yang_Zhou}.
The low energy dynamics in the specific construction therefore are characterized by a mix of two different dynamic critical exponents $z=1,2$. The low energy sector further contains gapless Lifshitz fermions, in addition to a single Weyl fermion.

Nevertheless,
one can remove these gapless Lifshitz points by increasing the TRB field,
\begin{equation}
 H_{B} (\mathbf{k}) = H (\mathbf{k})+  B \left( 2 - \cos(k_x)-\cos(k_y) \right) \sigma^y,
 \label{eq: mag_gapping_2}
\end{equation}
where $B$ is the amplitude of the magnetic coupling and is set to be larger than unity, i.e. 

\begin{equation}
B>1.   
\end{equation}

We have a conserved charge that doesn't vanish at $\mathbf{k} = 0$:
\begin{equation}
\label{eq: charge_maggap2}
    Q(\mathbf{k}) = -\frac{1}{2}\left((1+\cos(k_z))\tau^y - \sin(k_z) \tau^z\right)
\end{equation}

The charge operator flows to the usual chiral charge operator $-\tau^y$ in the continuum limit.
This charge is also similar to the sum of the axial and U(1) charges for the 1D staggered fermion model studied recently in \cite{Chatterjee_2025}. 



\section{Organizing lattice models in a family: An equivalence class of single Weyl fermion models}

\label{part: family}

We now further show that the infrared limit for the models in the different discussions can be organized into a family of Hamiltonians which transform in a $3\otimes 3$ dimensional linear representation of $\text{Spin}(4)$  group \cite{Zhou1} where one of the $SU(2)$ subgroups can be identified as a subgroup of the emergent Lorentz $SO(3,1)$ group.\\

The unitary group acting on these real fermions here is $\text{Spin(4) = SU(2)}\times $SU(2), so as we are considering theories of real fermions, we should consider the action of Spin(4) on the real fermion Hamiltonians studied. the spin(4) algebra is generated by two su(2) algebras that commute with each other, we choose a particular realization of their generators:

Let us name the first su(2) algebra:

\begin{equation}
    \Sigma = \text{Vect}(\tau_x\otimes\sigma_y , \tau_y\otimes\ I , \tau_z\otimes\sigma_y)
    \label{eq:Sigma}
\end{equation}

and the second:

\begin{equation}
    \Theta = \text{Vect}(\tau_y\otimes\sigma_x,I\otimes\sigma_y,\tau_y\otimes\sigma_z)
\label{eq:Theta}
\end{equation}

We can easily check that $\Sigma$ commutes with $\Theta$.

We first study the action of $\Sigma$ on different representations of Spin(4). Let us determine it's action on 
\begin{equation}
    V = (-\tau_z\otimes\sigma_z, -I\otimes\sigma_x,\tau_x\otimes\sigma_z)
\end{equation}

We notice that \begin{equation}
    [\Sigma_i,V_j] = 2i\epsilon^{ijk}V_k
\end{equation}
 and also that \begin{equation}
     \Sigma_iV_j = i\epsilon^{ijk}V_k
 \end{equation}

So the action of $\Sigma$ on $V$ is isomorphic to that of the algebra of $\sigma$-matrices on itself (the adjoint representation or the Pauli matrices algebra), which tells us from our knowledge of the latter that, if we let $U(\phi,\vec{n}) = e^{-i \frac{\phi}{2}\vec{n}\cdot\Sigma} = \cos(\frac{\phi}{2}) -i \sin(\frac{\phi}{2})(\vec{n}\cdot\Sigma)$ and $R(\phi,\vec{n})$ be the rotation of angle $\phi$ and axis $\vec{n}$ with $\lVert \vec{n} \rVert = 1$, then, let $\vec{v} \in \mathrm{R}^3$:

\begin{equation}
    U^\dagger(\phi,\vec{n})(\vec{v}\cdot V)U(\phi,\vec{n}) = (R(\phi,\vec{n})(\vec{v}))\cdot V
\end{equation}

$\Sigma$ has the same commutation relations if we replace $V$ by:

\begin{equation}
    V_2 = (\tau_z\otimes\sigma_x, -I\otimes\sigma_z,-\tau_x\otimes\sigma_x)
\end{equation}

or
\begin{equation}
    V_3 = (\tau_x\otimes I, \tau_y\otimes\sigma_y,\tau_z\otimes I)
\end{equation}
So it has the same action on the vector space they generate.

This tells us that the following Hamiltonian (for p-wave superconductors or superfluids) generates under $\Sigma$ a 3-dimensional representation of SU(2):

\begin{equation}
\mathcal{H}_{p.w.} = \int \frac{d\mathbf{r}}{2}\left[\chi^T(\mathbf{r})\left(\tau_z\otimes(\sigma_x i\nabla_z - \sigma_z i\nabla_x) + \tau_x\otimes\mathbb{I} i\nabla_y\right)\chi(\mathbf{r}) \right]
\end{equation}

Let us now study the action of $\Theta$ on this Hamiltonian:
let $W = (\tau_z\otimes \sigma_z, \tau_x\otimes I,\tau_z\otimes \sigma_x)$, the commutation relations are:
\begin{equation}
    [\Sigma_i,W_j] = -2i\epsilon^{ijk}W_k,
\end{equation} and we also have $$\Sigma_iW_j = -i\epsilon^{ijk}W_k$$

We let
$U'(\phi,\vec{n}) = e^{-i \frac{\phi}{2}\vec{n}\cdot\Theta}$, then we have:

\begin{equation}
    U'^\dagger(\phi,\vec{n})(\vec{v}\cdot W)U'(\phi,\vec{n}) = (R(-\phi,\vec{n})(\vec{v}))\cdot W
\end{equation} 

so 

\begin{equation}
    \begin{aligned}
    U'^\dagger (\phi,\vec{n})\vec{p}\cdot(\vec{v}\cdot W)U'(\phi,\vec{n}) &= \vec{p}\cdot(R(-\phi,\vec{n})(\vec{v}))\cdot W\\ 
    &= (R(\phi,\vec{n})(\vec{p}))\cdot\vec{v}\cdot W
\end{aligned}
\end{equation}

Therefore $\Theta$ acts on the momentum part of the Hamiltonian by a vector rotation, so $\mathcal{H}_{p.w.}$ also transforms in a 3-dimensional representation of SU(2) under $\Theta$. 

The $\Theta$ transformation is just a reformulation of Lorentz symmetry and thus can be offset by a Lorentz transformation.

In the following, we will show these actions of $Spin(4)$ explicitly in a tensor product representation. We will also present two dual copies of such representations and illustrate that all chiral lattice fermion models introduced in Sect.\ref{part: realFermion},\ref{part: Model_analysis},\ref{part: Model_analysisII} can be organized into these two representations and form an equivalence class.

\subsection{Spin(4) action in a tensor product representation}

We introduce $\mathcal{H} = \begin{pmatrix}
    V\\
    V_2\\
    V_3
\end{pmatrix}^T$.

Then a general Hamiltonian can be rewritten as:

\begin{equation}
    H = v^{\alpha}p^{\beta}\mathcal{H}_{\alpha\beta} =: (v \otimes p)\cdot\mathcal{H},
    \label{eq:3x3}
\end{equation}

or its dual copy

\begin{equation}
    {H} = p^{\alpha}v^{\beta}\mathcal{H}_{\alpha\beta} =: (p \otimes v)\cdot\mathcal{H}.
    \label{eq:3x3dual}
\end{equation}
Two copies of the Hamiltonians are related via a $\tau$-$\sigma$ transformation\cite{Zhou1}, which physically represents a charge ($\tau$)-spin ($\sigma$) duality.

Consequently we can rewrite the action of Spin(4): let $U = U_1U_2$ with $U_1$ (resp.$U_2$) being in the SU(2) group generated by $\Sigma$ (resp.$\Theta$); then, with the previous properties, the tensor product  of the representation of Spin(4) on v-space and on p-space writes:

\begin{equation}
    (Uv \otimes Up)\cdot\mathcal{H} = (U_1v \otimes U_2 p)\cdot \mathcal{H}=UHU^{-1}
\end{equation}
which is precisely the action described above.

Therefore, $H$ precisely transforms in the (1,1) representation of $Spin(4)$. $\bar{H}$ is in a dual of this representation.
All together, they form two dual copies of $9$-dimensional 
representations, i.e.

\begin{equation}
 H \in  (1,1)\oplus (\bar{1},\bar{1})
 \label{eq:dualRep}
\end{equation}
where the bars refer to the $\tau-\sigma$ charge-spin dual transformation.

For example, we have, at a superconducting tQCP,
\begin{equation}
    \mathcal{H}_{p.w.}=\left( \begin{pmatrix}
    1\\
    0\\
    0
\end{pmatrix} \otimes (-\vec{p})\right)\cdot\mathcal{H}.
\label{eq:1-1representation}
\end{equation}
It tranforms in one of the $9$-dimensional $(1,1)$ representations of $Spin(4)$. And if we only consider intrinsic transformations, not those acting on momenta, $\mathcal{H}_{p.w.}$ transforms in the 3-dimensional representation of SU(2).

\subsection{The mass operators in lattice models}

Let us further mention that one can easily further include the lattice mass operators in the above Hamiltonian. The mass operators $M({\bf k}) \Gamma_4$ defined in Sect.\ref{part: realFermion}C in Eq.26 for real fermions can only be among the ones in the algebraic group $Spin(4)$. That is,

\begin{equation}
    \Gamma_4 \in \{\Sigma_i, \Theta_j, i, j=1,2,3\}.
\end{equation}
They are simply the generators of two $su(2)$ subgroups of $spin(4)$ defined in Eq.\ref{eq:Sigma},\ref{eq:Theta}. They themselves form a $3\oplus 3$ adjoint representations of $Spin(4)$ group spanning a six dimensional space.

For instance, we can define $m({\bf k})\Gamma_4$ as

\begin{eqnarray}
   m({\bf k})\Gamma_4 =u_\alpha {\Sigma}_\alpha,\alpha=1,2,3.
\label{eq:gamma4}
\end{eqnarray}
It transforms in a three-dimensional representation of the $Spin(4)$ group; more precisely, it will be in the $(1,0)$
representation. 

If we instead define $\Gamma_4$ in terms of $\Theta_\alpha, \alpha=1,2,3$,
it forms a $\tau-\sigma$ dual of what we have above and still form a 
three-dimensional representation, $(0,1)$. Here,
$(0,1)$ forms a $\tau-\sigma$ dual of $(1,0)$. All the mass operators therefore form the following $6$-dimensional representation of the $Spin(4)$ group

\begin{eqnarray}
  \Gamma_4 \in (1,0) \oplus (0,1).
  \label{eq:massRep}
\end{eqnarray}

Time reversal symmetry will put further constraints on the choice of the mass operators in Eq.\ref{eq:gamma4}.

\subsection{The time reversal symmetry breaking fields}

The operator associated with the T-symmetry break actions was introduced in Sect.\ref{part: realFermion}C in Eq.30. These operators also have to be 
among the six operators of the $Spin(4)$ algebraic group and need to be in a dual represention of the mass operators so to satisfy the algebaric relations stated in Eq.31. For the mass operators defined in Eq.\ref{eq:gamma4}, 

\begin{eqnarray}
   B({\bf k})\Gamma_5 =w_\alpha {\Theta}_\alpha,\alpha=1,2,3.
\label{eq:gamma5}
\end{eqnarray}

Not surprisingly, they form a dual to the representation in Eq.\ref{eq:massRep}. That is all the $T$-symmetry breaking actions also form a six dimensional representation of $Spin(4)$, 
\begin{eqnarray}
 \Gamma_5 \in  (0,1) \oplus (1,0).
  \label{eq:magRep}
\end{eqnarray}

Note that in the practical constructions, $\{ u_\alpha, w_\alpha, \alpha=1,2,3\}$ are all constrained by $\{ v_\alpha, \alpha=1,2,3\}$ and the protecting symmetries. They can be uniquely set by the choice of $\{v\}$, so they don't lead to extra dimensions in the Hamiltonian manifold. 

Therefore, the whole lattice model structure can be simply encoded in the substructure of $H$, the effective Hamiltonian defined in Eq.\ref{eq:3x3},\ref{eq:3x3dual}.
In conclusion,
the lattice chiral fermions with a single weyl cone can therefore form an equivalence class defined by Eq.\ref{eq:3x3},\ref{eq:3x3dual},\ref{eq:dualRep},\ref{eq:gamma4},\ref{eq:massRep},\ref{eq:magRep},\ref{eq:1-1representation}. 

It is worth noting that the most general lattice models are defined in 
a $15\oplus 15$-dimensional manifold. The manifold is spanned
by a $15$-dimensional representation of 
a $Spin(4)$ group and its dual as defined below

\begin{eqnarray}
 H &\in & (1,1)\oplus (1,0)\oplus(0,1)  \mbox{  or}  \nonumber \\
 H &\in & (\bar{1},\bar{1})\oplus (0,1)\oplus(1,0).
 \label{eq:totalRep}
\end{eqnarray}

However, all the chiral fermion lattice models with a single Weyl fermion, after all symmetry constraints along with the Clifford algebra are imposed can be uniquely encoded in two dual copies of the much smaller 9-dimensional $(1,1)$ representations defined in Eq.\ref{eq:3x3},\ref{eq:3x3dual}. So effectively, the $15\oplus 15$ dimensional manifold can be projected into a $9\oplus 9$ dimensional sub-manifold where chiral fermion lattice models with single Weyl fermions discussed here all belong to.

Evidently, in the infrared limit where the single Weyl fermion emerges, the Hamiltonians simply also form two charge-spin dual copies of $3\otimes3$ dimensional representations of the $Spin(4)$ group,
 $$  H  \in  (1,1)\oplus (\bar{1},\bar{1})$$
where the bars refer to the $\tau-\sigma$ or charge-spin dual transformation.
In a sense, its infrared substructures contains unique information of the whole family of lattice models. And all models form an equivalence class either defined by a tQCP in the DIII class topological superconductors or by its dual, a nodal point phase, i.e. they are isomorphic.

\section{Generalized UV completed symmetries and symmetry charge operators}
\label{part: Charge}

This section is devoted to a detailed analysis of the UV completed  symmetry groups or charges in our models. Here we study the algebras and constraints of the commuting space of the Hamiltonian, which allows us to define and identify symmetry charges in different models.

Symmetry charges or charge-operators have distinct structures and span distinct spaces in models that are either non-degenerate along path {\bf b}), path \textbf{c}), i.e. the models in Eq.\ref{eq:H_p-wave},Eq.\ref{eq: tQCP2B},Eq.\ref{eq:mag_gapping_1},Eq.\ref{eq: mag_gapping_2} or degenerate along path {\bf a}), i.e. the model in Eq.\ref{eq: tQCP}.
We show that the space of conserved charges is larger in the latter. 

Below We characterize the algebraic structure of these symmetries, discuss their non-compact nature, and illustrate how they act non-locally on lattice degrees of freedom. 

In the following discussion, we characterize the structure of the space of possible exact charges in the context of $T$-symmetry breaking 
non-degenerate models (path {\bf b}), path {\bf c}) ) or $T$-symmetric degenerate models (path {\bf a)}).

\subsection{Charges for ${\cal T}$-symmetry breaking Hamiltonians}

For the first case, consider for example Hamiltonians in Eq.\ref{eq:H_p-wave}, \ref{eq: tQCP2B},\ref{eq:mag_gapping_1},\ref{eq: mag_gapping_2}. At nearly all points or any generic point in the Brillouin zone, their eigenvalues are non-degenerate, this indicates that any charge operators that commute with $H$ must obey:

\begin{equation}
    Q(\mathbf{k}_0) = \text{Diag}(s_1,s_2,s_3,s_4)
\end{equation}
in the basis that diagonalizes $H$.

This spans a 4 dimensional vector space to which $H$ itself and $I$ belong. So there are only two linearly independent rays left to construct possible symmetry charges if we project away the subspace spanned by $\{ H, I\}$. See Appendix F for the detailed constructions.

The two linearly independent symmetry charges can have the general structures,  
    $$Q_{I} (\mathbf{k}_0) = \text{Diag}(s_{I1},s_{I2},s_{I3},s_{I4}),$$
$$Q_{II} (\mathbf{k}_0) = \text{Diag}(s_{II1},s_{II2},s_{II3},s_{II4})$$
in the basis that diagonalizes $H$ at a specific momentum $\mathbf{k}_0$.
Furthermore, they commute with each other at a given momentum $\mathbf{k}_0$,

\begin{equation}
    [Q_I(\mathbf{k}_0), Q_{II}(\mathbf{k}_0]=0.
\end{equation}
Locally at a given $\mathbf{k}_0$ point, these symmetry charges carry the same algebra as $u(1)\oplus u(1)$.

However, generally charge operators in 
\begin{equation}
\{ Q_{I}({\bf k}), Q_{II}({\bf k}); {\bf k} \in \mathbb{T}^3 \}
\end{equation}
 defined at different points of $\mathbb{T}^3$ are $\mathbf{k}$- dispersive and also contain the null operator. And these charges defined above form rays spanning $\mathbb{R}\otimes \mathbb{R}$ (see below) rather than in a compact manifold.

These general observations suggest the non-compactness of the UV completed symmetry charges.
And the non-compactness of symmetries were emphasized in general studies of chiral fermions as well as in the explicit construction of lattice chiral fermions\cite{Friedan:1982nk,gioia2025exactchiralsymmetries31d} and they also appear in all of our studies here.

However, in the current studies of 3d single Weyl fermion lattice models here we further find that the dimension spanned by charge operators is {\it two}, higher than one. This suggests that the symmetry charges are not unique and there are multiple non-compact, non-abelian symmetry charges {\it even for a single Weyl fermion and their lattice models}. This is a surprising new feature that hadn't been investigated before, and needs to be fully explored in the future.

But it is also because of the two-dimensional linear space, one can always freely choose to work with a symmetry charge of a specific property. 
Following the discussions in Appendix F, it is indeed always possible to assign a single identical charge to any two different eigen states of $H$. Here we can apply this idea more explicitly in a concrete limit.

For example, we can zoom in and focus on some given $\mathbf{k}_0$ where 
$H(\mathbf{k})=\text{Diag}(1,0,0,-1)$ (in {\it arbitrary units}).
The two symmetry charges then can have the simple form $Q_I(\mathbf{k}) = \text{Diag}(0,1,-1,0)$ and $Q_{II} (\mathbf{k}) = \text{Diag}(-1,1,1,-1)$ in the basis that diagonalizes $H(\mathbf{k})$.

A charge operator must therefore be of the form (up to a unitary transformation): 
\begin{equation}
    Q(\mathbf{k}) = f(\mathbf{k})Q_{II}+g(\mathbf{k}) Q_{I}(\mathbf{k})+h(\mathbf{k})H(\mathbf{k})+l(\mathbf{k})I
\end{equation}
near this momentum of $\mathbf{k}_0$ where $f,g,h,l$ are smooth functions.

Let us now apply this idea to the nodal phase: only the two middle bands cross to form Weyl points.
At these Weyl points $\mathbf{k}_{deg}$, one verifies that $H(\mathbf{k})=\text{Diag}(1,0,0,-1)$ (in {\it arbitrary units}).

The charge operator should attribute the same value of charges to the two middle bands to reproduce the \textit{chiral charge} in the IR. This then further restricts our choice of a charge operator, leading us to choose only operators in the ray $\mathbb{R}\text{Diag}(-1,1,1,-1)$. In terms of $f,g,h,l$, it means that: $f(\mathbf{k}_{deg}) = \pm1$ and $g(\mathbf{k}_{deg}),h(\mathbf{k}_{deg}),l(\mathbf{k}_{deg}) = 0$.

Of course this constraint only applies in the vicinity of the crossings, but this motivates our choice (Eq.\ref{eq: Q_pwave}) as a suitable charge operator. 

Moreover, as discussed in Sect. \ref{part: realFermion}, some further constraints need to be further applied to these functions. The charge has to obey the charge conjugation symmetry $Q(\mathbf{k}) = -Q^*(-\mathbf{k})$, which in the case of charge (Eq.\ref{eq: Q_pwave}) imposes a constraint on the function $f$. The multiplying function $f$ has to be an odd function of $\mathbf{k}$, as for example $\sin(k_z)$. As depicted in FIG. \ref{fig:charge_nodal2}, it especially forces the attribution of an opposite charge to the two crossings that are related by charge conjugation symmetry, which is what we have wanted in the first place, as we are looking for a single complex Weyl cone.

It is worth emphasizing that the charge has to be 0 either at $\mathbf{k}=0$ or at one of the 8 time-reversal invariant points. 
It forces the symmetry group to be non-compact (because the spectrum necessarily has 2 values of irrational ratio) and therefore the charge to be non-on-site.

The case of the tQCP, which can be viewed as an extreme case of the Nodal phase (FIG. \ref{fig:Nodal_vs_tQCP2}), can actually evade the constraint of $S(0) = 0$ because the two crossings happen at the same point (see the previous Section). Two crossings at the same point of the momentum space transform into each other under the time reversal transformation; so two finite but opposite charges, $\pm 1$ can be assigned to each crossing. 

The finding that the charge for a single Weyl fermion cannot be quantized can be demonstrated in a more general fashion. This is done in Appendix \ref{part: charge cannot be quantized}.

\subsection{Charge for $\mathcal{T}$-symmetry invariant Hamiltonian}

In the second case, the Hamiltonian is two-fold degenerate at every point in the Brillouin zone, which is typically the case when there is Time-Reversal Invariance like in the model (Eq.\ref{eq: tQCP}). This allows us to define a much wider variety of symmetry charges or operators.

In particular, in a basis that diagonalizes $H(\mathbf{k})$, every operator of the following form commutes with the Hamiltonian:

\begin{equation}
    Q(\mathbf{k}_0) = \begin{pmatrix}
    A&0\\
    0&B
\end{pmatrix}
\end{equation}
With A and B being $2\times2$ hermitian matrices acting on the degenerating subspace of $H$. That is $S$ has a block-diagonal structure.

The symmetry charges span an 8 dimensional vector space to which H and I belong, so there are at least 6 linearly independent charges to look for. In terms of the sigma and tau matrices in the basis that diagonalizes H, our charge takes the following form:

\begin{equation}
    Q(\mathbf{k}) = \vec{a}(\mathbf{k})\cdot (I\otimes\sigma) + \vec{b}(\mathbf{k})\cdot(\tau^z\otimes\sigma) + c(\mathbf{k})H(\mathbf{k})+d(\mathbf{k})I
\end{equation}
with $\vec{a}(\mathbf{k}),\vec{b}(\mathbf{k})\in \mathbb{R}^3$ and $$\tau^z\otimes\sigma := (\tau^z\otimes\sigma^x,\tau^z\otimes\sigma^y,\tau^z\otimes\sigma^z)^T.$$

As opposed to the first case, even at a given momentum \textbf{k}, the charges here do not always commute with each other.
But they still form a closed algebra: if $A$ and $B$ commute with H, trivially $[A,B]$ also commutes with H. 

One can even show that the elements in 
$$\{\vec{a}(\mathbf{k})\cdot (I\otimes\sigma) + \vec{b}(\mathbf{k})\cdot(\tau^z\otimes\sigma)|\quad\vec{a}(\mathbf{k}),\vec{b}(\mathbf{k})\in \mathbb{R}^3\}$$
form a closed algebra, since none of the commutation relations between those elements involve $H$ or $I$. Moreover, from its structure we see that the algebras of symmetry charges is isomorphic to $\text{spin(4)}=\text{su(2)}\oplus \text{su(2)}$ algebra, in contrast to the $u(1)\oplus u(1)$ algebra of the first case.

However, as in the first case, all the six linearly independent charge operators are dispersive in the momentum space $\mathbb{T}^3$ and can also contain the null operator.
Charge operators in
\begin{equation}
\{ Q_{\alpha}({\bf k}), {\bf k} \in \mathbb{T}^3; \alpha =1,...,6 \} 
\end{equation}
therefore span a six-dimensional linear space and they form rays in 
$\mathbb{R}^3\otimes \mathbb{R}^3$, i.e. in a non-compact manifold. 
Again, in general we expect symmetries are non-compact and, in this case with the time reversal-symmetry, are also non-abelian. The action of symmetry transformation is non-local.

\subsection{Action of the charge on an on-site operator}

The action of the symmetry operator $e^{-i\lambda Q}$ on an on-site operator $\chi_i$ is described by the ordinary differential equation:

\begin{equation}
    i\frac{\partial\chi_{\bf i}}{\partial\lambda} = [Q,\chi_{\bf i}]
\end{equation}

Let us examine for example the charge for the tQCP model (Eq.\ref{eq: tQCP}):
\begin{equation}
    \begin{aligned}
    \mathcal{Q} = \sum_{\bf i}&[\left(\frac{-1}{2}(\sum_{\hat{b}} \chi_{{\bf i}+\hat{x}+\hat{b}}+\chi_{{\bf i}+\hat{x}-\hat{b}} ) +3\chi_{{\bf i}+\hat{x}}\right) I\otimes\sigma^z \chi_{\bf i} \\&- \left(\chi_{\bf i}- \chi_{{\bf i}+2\hat{x}}\right) \tau^x\otimes\sigma^y\chi_{\bf i}]+h.c.
\end{aligned}
\end{equation}

this gives us:

\begin{equation}
    \begin{aligned}
    i\frac{\partial\chi_{\bf i}}{\partial\lambda} =& \left(\frac{-1}{2}(\sum_{\hat{b}} \chi_{{\bf i}+\hat{x}+\hat{b}}+\chi_{{\bf i}+\hat{x}-\hat{b}} ) +3\chi_{{\bf i}+\hat{x}}\right) I\otimes\sigma^z\\&+I\otimes\sigma^z\left(\frac{-1}{2}(\sum_{\hat{b}} \chi_{{\bf i}-\hat{x}-\hat{b}}+\chi_{{\bf i}-\hat{x}+\hat{b}} ) +3\chi_{{\bf i}-\hat{x}}\right)\\&- \left(\chi_{\bf i}- \chi_{{\bf i}+2\hat{x}}\right) \tau^x\otimes\sigma^y - \tau^x\otimes\sigma^y\left(\chi_{\bf i}- \chi_{{\bf i}-2\hat{x}}\right) 
\end{aligned}
\end{equation}

So we see a direct coupling with the neighboring sites, which is a sign that the symmetry is non-local: as $\lambda$ grows, the weight of the operator distributes to still further lattice sites.

\section{Conclusions}
\label{part: Conclusion}

To summarize, we have put forward a general approach towards lattice chiral fermions starting with a gapped fermionic SPT which is time reversal invariant but breaks the charge $U(1)$ symmetry, i.e. a topological superconductor protected by the time reversal symmetry. 

Our main results are summarized below:

I)
We have explored thoroughly three main paths toward lattice chiral fermions starting with a DIII class fully gapped superconducting SPT.

Path \textbf{a}) preserves the time reversal symmetry but requires generic topological quantum critical points across which the change of topological invariant takes its fundamental value of  $\delta N_w=2$. 

Path \textbf{b}) involves a magnetic coupling that breaks the time reversal symmetry but doesn't require any fine tuning and are more robust.

Path \textbf{c}) is a hybrid of Path \textbf{a}) and Path \textbf{b}) and can also be applied to the construction of a lattice fermion model. The hybrid approach involves tQCPs across which the change of topological invariant takes a value of an integer multiple of the fundamental values, i.e. $\delta N_w=2n$, $n=1,2,..$ but further subject to a time reversal symmetry breaking field. 

The tQCPs with n being larger than two can be also associated to multi-critical tQCPs as they are described by multiple copies of the fundamental theory of tQCP with $\delta N_w=2$ studied in Ref.\cite{Zhou1,Zhou2,Zhou3}.

This unified practical approach can be applied to reproduce emergent single Weyl cones and associated IR symmetries previously obtained in these gapless superfluids, either as stable phases or quantum critical points. We have further explored the UV completion of those IR symmetries.

III)
The family of chiral lattice models with single Weyl fermions that have been known to us so far can be encoded in two dual copies of $3\otimes 3$-dimensional linear representations, or more precisely can be encoded in two dual $(1,1)$ representations of a Spin(4) group where one of the $SU(2)$ subgroups can also be identified as a subgroup of the emergent Lorentz group $SO(3,1)$.
They form an equivalence class. In the infrared limit, they are isomorphic to either a tQCP with time reversal symmetry or its dual, a $T$-symmetry breaking superconducting nodal point phase.

IV) Furthermore, it is further utilized to pinpoint an intimate connection between a three spatial dimensional lattice chiral fermion model recently constructed in \cite{gioia2025exactchiralsymmetries31d}, with exact non-on-site non-compact symmetries, and real fermions that naturally appear in gapless superfluids or superconductors studied previously. We found that the proposal in the above reference is equivalent to applying a $T$-breaking action  to a multiple-critical $T$-symmetry protected tQCP where the change of topologies, $\delta N_w=8$. This quadruples the more generic tQCP fundamental value of $\delta N_w=2$. This recent construction above is one example along Path \textbf{c}), i.e. a hybrid approach of
Path \textbf{a}) and Path \textbf{b}).

V) Finally, we also illustrate the differences in the construction of UV completed symmetries along different paths. For a generic $T$-symmetric tQCP along path a), the conserved-charge operators span a six-dimensional linear space while for a $T$-symmetry breaking gapless state, operators span a two-dimensional linear space instead.

A few open questions that we plan to further look into in the future studies:

i) The conserved charge operators can span either a two-dimensional or a six-dimensional linear space at each ${\bf p}$-point in the 3-torus $\mathbb{T}^3$. The conserved charge construction in this article as well as in previous chiral fermion discussions doesn't appear to be unique. The role played by these large linear spaces and how they can be systematically and explicitly patched together in the space of $\mathbb{T}^3$ remain to be further studied. 

ii) The IR theory of the lattice model has an emergent symmetry that is subject to a T'Hooft anomaly. This hinted an intimate connection between the IR physics and the surface of a 4D SPT where gapless surfaces can be probed by gauge anomalies\cite{Zhou3}. However, the UV physics and symmetry are very different. The 3D lattice has a non-compact $\mathbb{R}$-symmetry due to the continuous spectrum of the conserved charges while the 4D SPT has a compact $U(1)$ symmetry with charge well quantized. 

This poses a fundamental question: to what extent can a holography approach be meaningful and useful for the understanding of lower dimensional gapless states? 
How relevant are the UV completed symmetry groups at the IR and is there a UV-IR mixing which appears in the low energy sector?  

iii) At last, one can also ask how to generalize these discussions to interacting models? Can these non-compact non-local symmetries be broken
spontaneously and how are they different from the standard SSB phenomena of typical compact groups?

We acknowledge the support of an NSERC(Canada)
Discovery Grant under contract No RGPIN-2020-07070. G.M. acknowledges support from the Institut Philippe Meyer.
One of the authors(FZ) acknowledges a discussion with Ryan Thorngren on lattice chiral fermions. This project is in part to fulfill the ENS Master's internship requirements
(G.M.).

\appendix

\section{Mathematical demonstration of the No-go theorem}
\label{A_proof no go}

Let \( M = \mathbb{T}^3 \times \text{Herm}(n) \) and \( H_0: \mathbb{T}^3 \rightarrow \text{Herm}(n) \).

We set \( S := \{ (p, H(p)) \mid p \in \mathbb{T}^3 \} \), which is clearly a 3-dimensional submanifold of \( M \).

We want to study the intersection of \( S \) with the subset \( \Sigma \) of \( \text{Herm}(n) \) consisting of exactly 2-fold degenerate Hermitian matrices.

We will see that $X = \mathbb{T}^3\times\Sigma$ is a submanifold of $M$, and that although it is not closed (but $S$ is) it allows for a similar definition and properties of the intersection number as in \ref{proof_nogo1}.

Let us first describe \( \Sigma \):

Let \( H \in \Sigma \). Then there exist \( \lambda, \lambda_3, \ldots, \lambda_n \in \mathbb{R} \) and \( U \in U(n) \) such that:

\[ H = U \Lambda U^{-1}, \quad \text{where } \Lambda = \operatorname{Diag}(\lambda, \lambda, \lambda_3, \ldots, \lambda_n) \]

\( U(n) \) is of dimension \( n^2 \), but \( H \) is invariant under the transformation:

\[ U \rightarrow U 
\begin{bmatrix}
U(2) & & \\
 & U(1) & \\
 & & \ddots & \\
 & & & U(1)
\end{bmatrix} \]

Thus, we can consider the quotient \( U \in \tilde{U} = U(n) / (U(2) \times U(1) \times \cdots \times U(1)) \), which has dimension \( n^2 - (4 + (n-2)) = n^2 - n - 2 \).

The matrix \( \Lambda \) has \( n - 1 \) independent parameters (since we fix one eigenvalue for degeneracy).

Now, it is straightforward to verify that the map
\[
\begin{cases}
\tilde{U} \times \mathbb{R}^{n-1} \rightarrow \text{Herm}(n) \\
(U, \Lambda) \mapsto U \Lambda U^{-1}
\end{cases}
\]
is an immersion.

Therefore, \( \Sigma \) has dimension \( (n^2 - n - 2) + (n - 1) = n^2 - 3 \).

Similarly, the set of Hermitian matrices of k-fold degeneracy has dimension $n^2-k^2+1$, so for $k \geq 3$, the codimension is greater than 8, which will play a role in a moment.

We defined \( X = \mathbb{T}^3 \times \Sigma \). Both \( X \) and \( S \) are orientable manifolds of complementary dimension, and as discussed in the main text, the intersection number \( \text{int}(S, X) \) counts the sum of contributions from all Weyl cones in the Brillouin zone.

Consider a homotopy \( F \colon [0,1] \times \mathbb{T}^3\rightarrow M \) defined by:
\[
F(s,\mathbf{k}) = (1-s)\left(\mathbf{k},H(\mathbf{k})\right) + Q(\mathbf{k},A),
\]
where
\[
A = \text{Diag}(1,2,\cdots,n)
\]

($F$ is continuous because S is compact).

Since \( F \) is a homotopy, and we will show that homotopies preserve intersection numbers just as in the closed submanifold case \cite{1974Dt/b}, we have:
\[
\text{int}(S, X) = \text{int}(\mathbb{T}^3 \times \{A\}, X) = 0,
\]
Because \( A \) has distinct eigenvalues and so does not intersect with $X$.

\section{Intersection number is invariant by homotopy}
\label{Part_A_Intersection}

The proof is a direct adaptation of the one in \cite{1974Dt/b} (Chapter 3 Paragraph 3), the only missing assumption here is that $X$ is not a closed submanifold. In \cite{1974Dt/b} this is used to prove that if $F$ is a homotopy between transversal submanifolds, $F^{-1}(X)$ is a compact oriented one-manifold. We will thus prove that, in our case, we can always choose $F$ so that $F^{-1}(X)$ is compact.

Since $[0,1]\times T^3$ is compact, we only need to prove that $F^{-1}(X)$ is a closed set. As $F$ is continuous, this is equivalent to prove that we can choose a closed neighborhood $V$ of the image of $F$ such that $V\cap X$ is closed, or equivalently that $F$ avoids any points of $\partial X$, the set $\Sigma_{>2}$ of $\geq3$ degeneracies. 

We will use the fact that $\Sigma_{>2}$ is of codimension greater than 6 (greater than 8 around $\geq 3$-fold degeneracies, and 6 around double 2-fold degeneracies):

Assume that some point $x = (\lambda_0,\mathbf{k}_0)\in [0,1]\times \mathbb{T}^3$ is in $\Sigma_{>2}$: then from dimensional considerations there is always a vector $\vec{v}$ orthogonal to $T_x(T^3 \times \Sigma_{>2})$ and $T_xF([0,1]\times \mathbb{T}^3)$ to "push" this point out of $\Sigma_{>2}$ (by using a bump function $B$ that is non-zero only in the vicinity of $0$ for example: $\tilde{F}(\lambda_0+\lambda, \mathbf{k}_0+\mathbf{k}) = F(\lambda_0+\lambda, \mathbf{k}_0+\mathbf{k}) + B(\lambda, \mathbf{k})\vec{v}$), by the orthogonality property and pythagoras theorem, this strictly adds distance between $S$ and $\Sigma_{>2}$ (Fig. \ref{fig:int-inv})

\begin{figure}[h]
    \centering
    \includegraphics[width=0.8\linewidth]{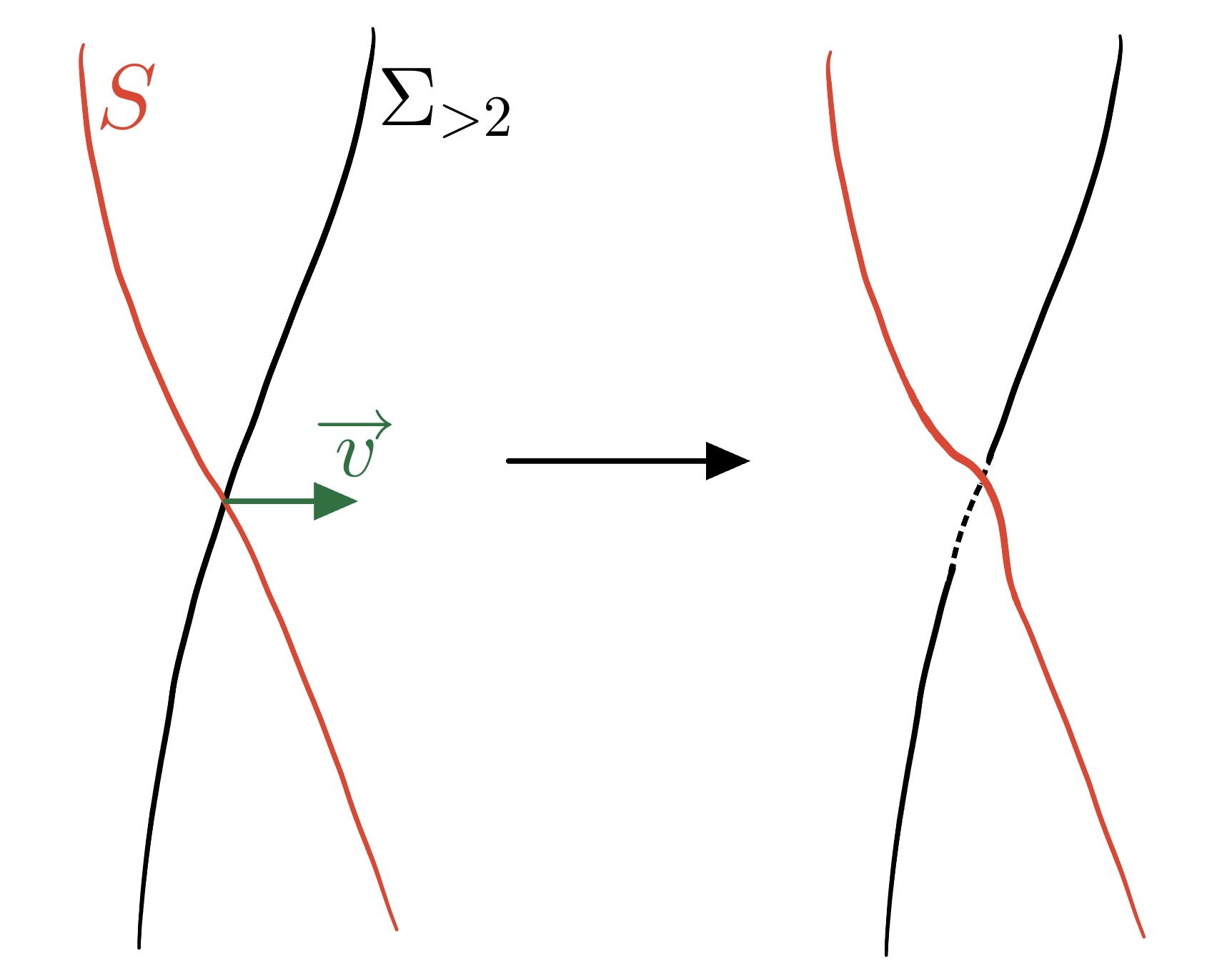}
    \caption{If $S$ and $\Sigma_{>2}$ intersect (left), we can find an orthogonal vector $\vec{v}$ to "push" $S$ away from $\Sigma_{>2}$ (right).}
    \label{fig:int-inv}
\end{figure}

This proves the invariance by homotopy and finishes this new proof of the Nielsen-Ninomiya theorem.

\section{Side results of our proof}

It can be noted that if instead of taking $\Sigma$, we choose $\Sigma_{[k,k+1]}$, the set of matrices with degeneracy only for band k and k+1, $\Sigma_{[k,k+1]}$ is still a submanifold, and the reasoning is the same. It thus gives us a more precise doubling theorem: the sum of chiralities of every Weyl points at crossing between band k and k+1 is 0. This could not be predicted from the proof by homotopy \cite{NIELSEN198120}.

\section{Chiral charge for a single Weyl fermion cannot be quantized}
\label{part: charge cannot be quantized}

This paragraph proves that even in the real fermmion formalism, no quantized chiral charge for a single Weyl fermion can be defined. 

Imagine we want to define a quantized charge (with values only 1 and -1) for the model corresponding to $F(0)$. For some Weyl point, let's say the charge is $+1$. Then, since the only values available are $\pm1$, there is only one way to define the charge (at least in a neighborhood of the Weyl points, because we don't allow 3-fold degeneracies) as a continuous fonction of $\lambda \in [0,1]$. Now \ref{thm: NN} tells us that at some $\lambda_0$, the Weyl point has to disappear with one other Weyl point of opposite chirality, so at that $\lambda_0$, the charge must also be $+1$ for the other Weyl point too. As the charge cannot change sign during the deformation, we deduce that that other Weyl point was also at charge $+1$ at parameter $\lambda=0$. This proves that no quantized charge for single Weyl fermion can be defined on a lattice.

\section{Applying the Schrieffer-Wolff transformation}
\label{appendix:SW}

We write the total Hamiltonian as:
\begin{equation}
    H(k) = H_0(k) + V_k
    \label{eq:SW_initial}
\end{equation}
where $H_0(k)$ is block-diagonal and $V_k$ contains the off-diagonal perturbations. 

The Schrieffer-Wolff transformation performs a unitary change of basis:
\begin{equation}
    H'(k) = e^{iS_k} H(k) e^{-iS_k}
\end{equation}
where \( S_k \) is a Hermitian operator chosen such that the transformed Hamiltonian \( H'(k) \) is block-diagonal to the desired order in perturbation theory.

Expanding using the Baker-Campbell-Hausdorff formula gives:
\begin{equation}
    H'(k) = H_0(k) + \frac{i}{2} [S_k, V_k] + \mathcal{O}(V_k^3)
    \label{eq:SW_result}
\end{equation}
where $S_k$ is the Hermitian generator of the transformation, defined by the condition:
\begin{equation}
    i[H_0(k), S_k] = V_k.
\end{equation}

Moreover, the Hamiltonian satisfies the symmetry condition:
\begin{equation}
    H(-k - k_{\mathrm{deg}}) = -H^T(k + k_{\mathrm{deg}})
    \label{eq:SW_symmetry}
\end{equation}

This implies that if \begin{equation}
    H_0(k + k_{\mathrm{deg}})
\end{equation} is block-diagonal, then so is 
\begin{equation}
    H_0(-k - k_{\mathrm{deg}}).
\end{equation}Furthermore, if \begin{equation}
    [S_{k + k_{\mathrm{deg}}}, V_{k + k_{\mathrm{deg}}}]
\end{equation} is block-diagonal, the same holds for 
\begin{equation}
    [S_{-k - k_{\mathrm{deg}}}, V_{-k - k_{\mathrm{deg}}}].
\end{equation} Consequently, the Schrieffer-Wolff transformation block-diagonalizes the Hamiltonian in the vicinity of both degeneracy points simultaneously.

In the vicinity of $+\mathbf{k}_{deg}$, the Hamiltonian takes the form:
\begin{equation}
    H(\mathbf{k}+\mathbf{k}_{deg}) = e^{iS_{\mathbf{k}+\mathbf{k}_{deg}}}\begin{pmatrix}
        \vec{A}(\mathbf{k}) \cdot \sigma & 0\\
        0 & *
    \end{pmatrix}e^{-iS_{\mathbf{k}+\mathbf{k}_{deg}}}
\end{equation}
with $\vec{A}(0)=0$.

If we linearize around $\mathbf{k}_{deg}$, we then have :
\begin{equation}
    \vec{A}(\mathbf{k}) = \mathbf{k}_\kappa V^\kappa
\end{equation}
for some invertible matrix $V$.
We can then project out the two blue dotted bands.

\section{Charge assignment}

We want to first prove that the spectrum for the Hamiltonian and the exact charge is flipped under charge-conjugation, and in the same way for the Hamiltonian and for the charge. Then we explore the two possibilities for the charges of an almost-everywhere non-degenerate Hamiltonian. And finally we illustrate this in some concrete examples.

\subsection{The spectrum is flipped under charge-conjugation}

We assume that we have a charge-conjugation symmetric Hamiltonian $H$ and an exact charge $Q(\mathbf{k})$ so that $[Q(\mathbf{k}),H(\mathbf{k})] = 0$ for all $\mathbf{k}$. The charge-conjugation symmetry is defined by an anti-unitary transformation
$\mathcal{C}$, with $\mathcal{C}^2=1$. It leads to the following relations,

\begin{subequations}
\begin{align}
    \mathcal{C} H(\mathbf{k})\mathcal{C}^{-1} &= -H(-\mathbf{k})\\
    \mathcal{C}Q(\mathbf{k})\mathcal{C}^{-1} &= -Q(-\mathbf{k})
\end{align}
\end{subequations}

Let us first prove that the Hamiltonian and Charge speCtrum is flipped under Charge-Conjugation symmetry:

Let $\ket{\alpha}_\mathbf{k}$ be an eigenstate of $H(\mathbf{k})$ and of $Q(\mathbf{k})$ (since they commute, we can always find such an eigenvector for each eigenvalue), we have:

\begin{subequations}
\begin{align}
H(\mathbf{k})\ket{\alpha}_\mathbf{k} &= E_\alpha(\mathbf{k}) \ket{\alpha}_\mathbf{k}\\
Q(\mathbf{k})\ket{\alpha}_\mathbf{k} &= s_\alpha(\mathbf{k}) \ket{\alpha}_\mathbf{k}
\end{align}
\end{subequations}
Then:
\begin{equation}
    H(-\mathbf{k})\mathcal{C}\ket{\alpha}_\mathbf{k} = -\mathcal{C}H(\mathbf{k})\ket{\alpha}_\mathbf{k} = -E_\alpha(\mathbf{k}) \mathcal{C}\ket{\alpha}_\mathbf{k}
\end{equation}
So $\mathcal{C}\ket{\alpha}_\mathbf{k} =\ket{\alpha}_{-\mathbf{k}}$ is an eigenvector of $H(-\mathbf{k})$ with eigenvalue $E_\alpha(-\mathbf{k}) = -E_\alpha(\mathbf{k})$.
The same is true for the charge, therefore $s_\alpha(-\mathbf{k}) = -s_\alpha(\mathbf{k})$ is the eigenvalue corresponding to the same exact eigenvector $\mathcal{C}\ket{\alpha}_\mathbf{k} =\ket{\alpha}_{-\mathbf{k}}$.

This in particular implies that if there is a band crossing at $+\mathbf{k}_0$, there is another one at $-\mathbf{k}_0$. And the charge eigenvalues related to the crossing eigen states at $+\mathbf{k}_0$ are translated into charges for the crossing eigen states at $-\mathbf{k}_0$, but with their sign flipped.

\subsection{The two independent charges for non-degenerate models}

$H$ is assumed to be non-degenerate almost everywhere, and the charges are assumed to be orthogonal to identity and $H(\mathbf{k})$. Algebraically, this translates into:
\begin{subequations}
    \begin{align}
        Tr(H(\mathbf{k})Q(\mathbf{k})) &= 0\\
        Tr(\mathbb{I}Q(\mathbf{k})) &= 0.
    \end{align}
\end{subequations}
The space of such charges evaluated at a particular $\mathbf{k}$-point describes a 2-dimensional vector space. We can choose 2 orthogonal charges in this vector space $Q_I(\mathbf{k})$ and $Q_{II}(\mathbf{k})$.

Then all linear combinations of $Q_I(\mathbf{k})$ and $Q_{II}(\mathbf{k})$ are valid choices for the charge.

One particular case that is of interest for the examples we considered in this paper is that we can always choose the charge to attribute the same eigenvalue to any 2 given bands:

let $Q(\mathbf{k}) = XQ_I(\mathbf{k})+YQ_{II}(\mathbf{k})$, we label $s_{I1},...,s_{I4}$ and  $s_{II1},...,s_{II4}$ the eigenvalues of $Q_I(\mathbf{k})$ and $Q_{II}(\mathbf{k})$ associated to the eigenstates $\ket{1},...,\ket{4}$, then the equation:
\begin{equation}
    X(s_{I1}-s_{I2}) + Y(s_{II1}-s_{II2}) = 0
\end{equation}
has at least one non-zero solution $(X_0,Y_0)$. This translates to:
\begin{equation}
    Xs_{I1}+ Ys_{II1} = Xs_{I2}+ Ys_{II2}
\end{equation}
so that the resulting charge $Q(\mathbf{k})$ has degenerate eigenvalues:
\begin{equation}
    s_1 = s_2.
\end{equation}

\subsection{concrete examples}

In section \ref{part: Model_analysis}, we determined a charge (Eq.\ref{eq: Q_pwave}) for the nodal phase:
\begin{equation}
    \label{eq: Q_pwave2}
    Q(\mathbf{k}) = \sin(k_z)\left[ \sin(k_z) \tau^x \otimes \sigma^y-( \mu - \sum_i \cos(k_i)) I \otimes \sigma^z\right]
\end{equation}
This charge has 2-fold degenerate eigenvalues for its two middle bands and lowest-highest bands (FIG. \ref{fig:charge_nodal2}) and the charge flips when passing from $+\mathbf{k}_{deg}$ to $-\mathbf{k}_{deg}$

\bibliography{references} 
\bibliographystyle{apsrev4-2}
\nocite{*}

\end{document}